\documentclass[sigconf]{acmart}
\AtBeginDocument{%
  }

\setcopyright{acmlicensed}
\copyrightyear{2025}
\acmYear{2025}
\acmDOI{XXXXXXX.XXXXXXX}
\acmConference[ISSAC'25]{the 2025 International Symposium on Symbolic and Algebraic Computation}{July 28-- August 1, 2025}{Guanajuato, Mexico}
\acmISBN{978-1-4503-XXXX-X/18/06}

\usepackage{amsmath}
\usepackage{amscd}
\usepackage{natbib}
\usepackage{amsfonts}
\usepackage{hyperref}
\usepackage{hyperxmp}
\newcommand{\bZ}{{\mathbb{Z}}}
\newcommand{\bQ}{{\mathbb{Q}}}
\newcommand{\bC}{{\mathbb{C}}}
\newcommand{\bN}{{\mathbb{N}}}

\newcommand{\vw}{{\bf w}}

\renewcommand{\vv}{{\bf v}}

\newcommand{\lc}{\operatorname{lc}}

\newcommand{\im}{\operatorname{im}}

\newcommand{\polypart}{\operatorname{poly}}
\newcommand{\proppart}{\operatorname{proper}}

\newcommand{\supp}{\operatorname{supp}}

\newcommand{\polylog}{\operatorname{polylog}}

\newtheorem{thm}{Theorem}[section]
\newtheorem{prop}[thm]{Proposition}
\newtheorem{cor}[thm]{Corollary}
\newtheorem{lemma}[thm]{Lemma}
\newtheorem{define}[thm]{Definition}

\newtheorem{alg}[thm]{Algorithm}
\newtheorem{convention}[thm]{Convention}
\newtheorem{remark}[thm]{Remark}
\newtheorem{example}[thm]{Example}

\begin{document}

\title{Complete Reduction for Derivatives in a Primitive Tower}

\titlenote{H.\  Du was partially supported by the NSFC grant (\# 12201065). 
The research of Yiman Gao was funded in part by the Austrian Science Fund (FWF) 10.55776/PAT1332123.
W.\ Li  and Z.\  Li  were partially supported by a National Key R\&D Program of China 2023 YFA1009401, the NFSC grant (\# 12271511)
and the CAS Fund of the Youth Innovation Promotion Association (No. Y2022001).\\
 For open access purposes, the authors have applied a CC BY public copyright license to any author-accepted manuscript version arising from this submission.}
%\thanks{H.\  Du was partially supported by an NSFC grant (\# 12201065). 
%The research of Yiman Gao was funded in part by the Austrian Science Fund (FWF) 10.55776/PAT1332123.
%W.\ Li  and Z.\  Li  were partially supported by a National Key R\&D Program of China 2023 YFA1009401, an NFSC grant (\# 12271511)
%and the CAS Fund of the Youth Innovation Promotion Association (No. Y2022001)
%
% For open access purposes, the authors have applied a CC BY public copyright license to any author-accepted manuscript version arising from this submission.}
%\titlenote{}
\author[H.\ Du]{Hao Du}
\affiliation{%
  \institution{School of Mathematical Sciences, \\ Beijing University of Posts and Telecommunications,}
  \city{Beijing (102206), China}
  \state{}
  \postcode{}
  \country{}
}
\email{haodu@bupt.edu.cn}

\author[Y.\ Gao]{Yiman Gao}
\affiliation{%
  \institution{Johannes Kepler University Linz, \\ Research Institute for Symbolic Computation (RISC),}
  \city{Altenberger Stra\ss e 69, 4040, Linz, Austria}
  \state{}
  \postcode{}
  \country{}
}
\email{ymgao@risc.jku.at}

\author[W.\ Li]{Wenqiao Li}
\affiliation{%
  \institution{Key Lab of Mathematics Mechanization, AMSS, \\ University of Chinese Academy of Sciences,}
  \city{Chinese Academy of Sciences, Beijing (100190), China}
  \state{}
  \postcode{}
  \country{}
}
\email{liwenqiao@amss.ac.cn}

\author[Z.\ Li]{Ziming Li}
\affiliation{%
  \institution{Key Lab of Mathematics Mechanization, AMSS, \\ University of Chinese Academy of Sciences,}
  \city{Chinese Academy of Sciences, Beijing (100190), China}
  \state{}
  \postcode{}
  \country{}
}
\email{zmli@mmrc.iss.ac.cn}

\renewcommand{\shortauthors}{H.\ Du, Y.\ Gao, W.\ Li and Z.\ Li.}

\begin{abstract}
A complete reduction $\phi$ for derivatives in a differential field is 
a linear operator on the field over its constant subfield.
The reduction enables us to decompose an element $f$ as the sum of
a derivative and the remainder $\phi(f)$. 
A direct application of $\phi$ is that $f$ is in-field integrable if and only if
 $\phi(f) = 0.$ 

In this paper, we present
a complete reduction for derivatives in a primitive tower algorithmically.
Typical examples for primitive towers are differential fields generated by (poly-)logarithmic functions and logarithmic integrals.
Using remainders and residues, we provide a necessary and sufficient condition for an element from a primitive tower 
to have an elementary integral, and discuss how to
 construct telescopers for non-D-finite functions in some special primitive towers.
\end{abstract}




\maketitle

\section{Introduction}\label{SECT:intro}

Let $V$ be a linear space and $U$ a subspace of $V$.
A linear operator $\phi$ on $V$ is called a {\em complete reduction} for $U$ if $v - \phi(v) \in U$ for all $v \in V$ and $U = \ker(\phi)$ 
%according to 
by \cite[Definition 5.67]{KauersBook}. Such an operator $\phi$ is an idempotent and results in
$V = U \oplus \im(\phi)$. 

Let $K$ be a differential field with derivation $^\prime$ and $C$ be the subfield of constants in $K$.
For $L \subset K$,  $L^\prime := \{ l^\prime \mid l \in L\}$. 
Then $L^\prime$ is a $C$-subspace. For  a complementary subspace $R$ for $K^\prime$, the
projection from $K$ to $R$ is a complete reduction for $K^\prime$. So there always exist complete reductions for $K^\prime$.
It remains %to construct such a complete reduction, that is,
\begin{enumerate}
\item to fix a complementary subspace $R$ for $K^\prime$, and
\item to develop an algorithm that, for every $f \in K$, computes $g \in K$ and $r \in R$ such that $f = g^\prime + r$.
\end{enumerate}
In general, both $K^\prime$ and $R$ are infinite-dimensional. 
\begin{example} \label{EX:rational}
Let $C$ be a field  of characteristic zero, and
$^\prime$ be the usual derivation $d/dx$ on $C(x)$.
A complementary subspace $R$ for $C(x)^\prime$ is the set of proper rational functions with squarefree denominators.
For  every $f \in C(x)$,
the Hermite-Ostrogradsky reduction on \cite[page 40]{BronsteinBook} computes $(g, r) \in C(x) \times R$ such that
$f = g^\prime  + r$.
The projection from $C(x)$ to $R$ is a complete reduction for $C(x)^\prime$.
%It is the first step towards the reduction-based creative
%telescoping for bivariate rational functions in \cite{BCCL2010}.
\end{example}

Our work is motivated by reduction-based creative telescoping (see \cite[\S 5.6]{KauersBook} and \cite[\S 15]{Salvy2019})
and integration (summation) in finite terms (see \cite{BronsteinBook,Risch1969,RS2022,Karr1981,Schneider2021,Schneider2023}). 
Both need preprocessors to split an integrand (summand) as the sum
of an integrable (summable) part and a possibly non-integrable (non-summable) part.

A commonly-used preprocessor in reduction-based creative telescoping is also known as 
 an additive decomposition, which can be described in terms of linear algebra  below.
 
  Let $V$ and $U$ be the same as those in the first paragraph.
For an element $v \in V$, an additive decomposition for $U$ computes $u \in U$ and $r \in V$ such that $v = u + r$, where $r$ is minimal in some sense.
And  $v \in U$ if and only if $r=0$.
It is proposed for constructing minimal telescopers in \cite{AP2001, AP2002, AGL2002, LeThesis}, in which $V$ is the $C(x,y)$-subspace spanned by a  hypergeometric term  in $x$ and~$y$,
and $U$ is the $C$-subspace $\{ g(x, y+1)-g(x,y) \mid g \in V\}$.
Additive decompositions also appear  in \cite{CDL2018,DGLW2020}, in which $V$ is a primitive tower of some special kinds,  and $U$ consists of all derivatives in $V$.

A complete reduction is interpreted  as an additive decomposition in \cite[\S 1.2]{GaoPhD2024} as follows.
Let $\phi$ be a complete reduction  for $U$ on $V$, $G$ be a basis of $U$, and $H$ be a basis of $\im(\phi)$. Then
$G \cup H$ is a basis of $V$. For every $v \in V$,  $v = \sum_{w \in G \cup H} c_w w$ with $c_w \in C$.
Define $\supp(v)=\{w \in G \cup H \mid c_w \neq 0\}$. For $v_1, v_2\in V$, we say that $v_1$ is not higher than $v_2$ if $\supp(v_1) \subseteq \supp(v_2)$.
If $v = u + r = \tilde{u} + \tilde{r}$ for some $u, \tilde{u} \in U$, $r \in \im(\phi)$ and $\tilde{r} \in V$,
then $\supp(r) \subseteq \supp( \tilde{r})$ by an easy linear-algebra argument. Thus, $r$ is not higher than $\tilde r$.

Additive decompositions do
not always induce linear maps. So they are not necessarily complete reductions. 
Since linearity  brings a lot of convenience into both theory and practice,
it is worthwhile to seek complete reductions. So far complete reductions have been developed
for hyperexponential functions \cite{BCCLX2013}, algebraic functions \cite{CKK2016,CDK2021}, fractions of differential polynomials \cite{BLLRR2016},
fuchsian D-finite functions \cite{CHKK2018}
and D-finite functions \cite{BCPS2018,vdHJ2021,CDK2023}.

%In this paper, we extend the complete reduction in Example \ref{EX:rational} to primitive towers.
%It leads naturally to  an algorithm for determining the in-field integrability of such towers,
%and can be applied to compute elementary integrals over the towers whose base fields are given in Example \ref{EX:rational}. 
%We also construct telescopers for non-D-finite functions of certain type by the reduction.

A classical topic in symbolic integration is to compute elementary integrals of transcendental Liouvillian functions (see \cite{BronsteinBook,Davenport1986,RaabThesis,RS2022,SSC1985}).
Results about this topic are usually described in monomial extensions (see \cite[\S 3.4]{BronsteinBook}). 

Let $K$ and $C$ be given  in the second paragraph, and $t$  be a monomial over $K$ (see \cite[Definition 3.4.1]{BronsteinBook}).  
The monomial extension $K(t)$ contains three $C$-subspaces highly relevant to integration. They are: 
$K(t)^\prime$ consisting of all derivatives in $K(t)$, $S_t$ consisting of proper fractions whose denominators are normal polynomials,
and $ W_t$ consisting of elements whose denominators are coprime with every normal polynomial (see \cite[Definition 3.4.2]{BronsteinBook}).  
%All of them are linear subspace over the subfield of constants in $E$.
Algorithm {\sc HermiteReduce} in \cite[\S 5.3]{BronsteinBook}
decomposes an element $f$ of $K(t)$ as the sum of a derivative, an element $s$ of $ S_t$ and an element $w$ of $W_t$.

Assume further that $t$ is either primitive or hyperexponential (see \cite[Definition 5.1.1]{BronsteinBook}) and that $K(t)$ and $K$ have the same constants.   
One tries to integrate $s$ by the residue criterion \cite[Theorem 3]{Raab2012}, and
$w$ by solving parametric Risch equations \cite[{\sc Main Theorem}]{Risch1969} and 
the parametric
logarithmic derivative problem \cite[\S 7.3]{BronsteinBook}. This approach results in  
an algorithm for deciding in-field integrability in arbitrary primitive towers (see Definition \ref{DEF:tower}). The algorithm may be turned
into an additive decomposition.
  
To develop a complete reduction, we take a different approach to handling  elements in $ W_t$. 
The approach proceeds in three steps:
\begin{enumerate}
\item[1.] Define an auxiliary subspace $A$ such that $W_t = W_t^\prime + A$.
\item[2.] Determine a basis of $W_t^\prime \cap A$.
\item[3.] Fix a complement  of $W_t^\prime$ in $W_t$ by the above basis.
\end{enumerate}
The  projection from $W_t$ to the complement  is a complete reduction for $W_t^\prime$,
which, together with Algorithm {\sc HermiteReduce}, leads to a complete reduction for derivatives in $K(t)$. %a monomial extension.
%The above-mentioned subspaces are infinite-dimensional in most cases.

We prefer to work out all the details for the case, in which $t$ is a primitive monomial,
although our approach is likely applicable to other cases (see \cite{GaoPhD2024,CDGHLL2025}).
This is because the approach for primitive monomials
does not  lead to any complicated case distinction, which seems unavoidable in other cases (see \cite[Section 4.1]{BCCLX2013}).

The auxiliary reduction (Algorithm \ref{ALG:ar}) developed in step 1 and construction of a basis for $ W_t^\prime \cap A$ in step 2
benefit from the way of using integration by parts to reduce polynomial integrands
in \cite{CDL2018,DGLW2020}, while the key lemma (Lemma \ref{LM:basis})  for step 2 is based on
not only integration by parts but also 
the fact that the parametric Risch
equation  in our case is of the form $y^\prime = c t^\prime + a$, where 
$a \in K$ is given, and $(y,c) \in K \times C$ is to be determined. 
If a complete reduction $\phi: K \rightarrow K$ for $K^\prime$ is available, then $t^\prime = u^\prime + \phi(t^\prime)$ and 
$a = v^\prime + \phi(a)$ for some $u,v \in K$.   
An application of $\phi$ to the above equation
yields $c \phi(t^\prime) + \phi(a) = 0$. 
Thus, $c$ is determined, and $y$ can be taken as $c u +v$ when $\phi(a) \phi(t^\prime)^{-1}$ is a constant.  
There is no need to solve any limited integration problem \cite[\S 7.2]{BronsteinBook}.
Algorithm \ref{ALG:proj} developed  in step 3 is a dual technique for representing 
a subspace by the intersection of kernels of linear functions.

In this paper, we develop a complete reduction for derivatives in primitive towers by the above approach.
%Reduced elements are polynomials in this case. %is a polynomial algebra.
The reduction leads to  an algorithm for determining in-field integrability (see Examples \ref{EX:rattower} and \ref{EX:alg}),
and can be applied to compute elementary integrals over such towers (see Example \ref{EX:elem}). 
We also construct telescopers for some non-D-finite functions by the reduction
(see Example \ref{EX:telesoper}).

%Our idea is also different from that for the additive decomposition in S-primitive towers \cite{DGLW2020}, although
%both make essential use of integration by parts to reduce polynomial integrands.
%In addition, primitive towers include S-primitive ones as a special case.

The rest of this paper is organized as follows.  In Section \ref{SECT:pre},
we specify notation and present several algorithms to be used in the sequel.
%These algorithms are either known or straightforward.
Basic constructions in the above three steps are described in Section \ref{SECT:basic}.
The constructions yield an algorithm for our complete reduction, as soon as the notion of primitive towers is introduced in Section \ref{SECT:tower}.
Some applications of the complete reduction are presented in Section \ref{SECT:app}.
Concluding remarks are given in Section \ref{SECT:conc}. 

\section{Preliminaries}\label{SECT:pre}

This section has  three parts.  In Section \ref{SUBSECT:notation}, we introduce some basic notions concerning symbolic integration and
fix notation to be used.
Effective bases are defined and constructed in Section \ref{SUBSECT:basis}.
They allow us to apply a dual technique in linear algebra.
In Section \ref{SUBSECT:log}, we review an algorithm in the proof of \cite[Theorem 3.9]{RaabThesis}, which helps us 
compute elementary integrals in Section \ref{SECT:app}.

\subsection{Notation and rudimentary notions} \label{SUBSECT:notation}

Throughout the paper, $G^\times$ denotes $G \setminus \{0\}$ for an additive group $(G, +, 0)$.
For $n \in \bN$,
the sets $\{1, \ldots, n\}$ and $\{0, 1, \ldots, n\}$ are denoted by
$[n]$ and $[n]_0$, respectively. The transpose of a matrix is denoted by $( \cdot )^\tau$.
Comments in an algorithm are placed between $(^*$ $\cdots$ $^*)$.

All fields are of characteristic zero in the paper. Let $K$ be a field.  We denote its algebraic closure by $\overline{K}$.
For a univariate polynomial $p$ over $K$, its degree and leading coefficient
are denoted by $\deg(p)$ and $\lc(p)$, respectively, when the indeterminate is clear from context.
In particular, $\deg(0):=-\infty$ and $\lc(0):=0$. Similarly, a univariate rational function
is said to be {\em proper} if the degree of its numerator is less than that of its denominator.
A rational function $r$ can be uniquely written as the sum of a polynomial and a proper rational function,
which are denoted by $\polypart(r)$ and $\proppart(r)$, respectively.

%Let $V$ be a $K$-linear space and $v_1, \ldots, v_m \in V$.
%The subspace spanned by $v_1, \ldots, v_m$ over $K$ is denoted by $\spa_K\{v_1, \ldots, v_m\}$.

A map $^\prime:$ $K \rightarrow K$ is called a {\em derivation} on $K$
if $(a+b)^\prime = a^\prime + b^\prime$ and $(ab)'=ab'+a'b$ for all $a, b \in K$.
A {\em differential field} is a field equipped with a derivation.
Let $(K, \, ^\prime)$ be a differential field.
An element $c$ of $K$ is called a {\em constant} if $c^\prime=0$.
All constants in $K$ form a subfield.
%The derivation is a linear map
%over the subfield of constants.
%Let $L \subset K$. We set $L^\prime = \{ l^\prime \mid l \in L\}$.
%In particular, $K^\prime$ is a subspace over the subfield of constants. %consisting of all derivatives in~$K$.
%It is a subspace over the subfield of constants.
A differential field $(E, \delta)$ is called a {\em differential field extension} of $(K, \, ^\prime)$
if $K$ is a subfield of $E$ and $^\prime$ is the restriction of $\delta$  to $K$.
We still use $^\prime$ to denote $\delta$  when there is no confusion.

Assume that $t$ belongs to a differential field extension of $K$.
If $t$ is transcendental over $K$ and $t^\prime \in K[t]$,
then $t$ is called {\em a monomial} over $K$ and $K(t)$ is called a {\em monomial extension} of $K$.

Let $t$ be a monomial over $K$. A polynomial $p \in K[t]^\times$ is said to be {\em normal}
if $\gcd(p, p^\prime)=1$. %Normal polynomials are squarefree.
An element $f$ of $K(t)$ is said to be {\em simple} if it is proper and has a normal denominator.
The subset consisting of all simple elements is denoted by $S_t$,
which is a $K$-subspace.
Note that $f$ is simple if it has a normal denominator in \cite[Definition 3.5.2]{BronsteinBook}. We further require that $f$ is proper
for the uniqueness of $s$ in \eqref{EQ:hr} given below. 
We call $t$ a {\em primitive monomial} over $K$ if $t'\in K \setminus K'.$
A primitive monomial extension $K(t)$ has no new constant other than the constants in $K$
by \cite[Theorem 5.1.1]{BronsteinBook}.

Let  $t$ be a primitive monomial over $K$. %Then $K[t]$ is a differential $K$-algebra.  
%A polynomial in $t$ is normal if and only if it is squarefree by \cite[Theorem 5.1.1]{BronsteinBook}.
For every $f \in K(t)$, there exists $g \in K(t)$, $p \in K[t]$ and a unique $s \in S_t$
such that
\begin{equation} \label{EQ:hr}
f = g^\prime + p + s.
\end{equation}
The uniqueness of $s$ is due to \cite[Lemma 2.1]{CDL2018}.
Algorithm {\sc HermiteReduce} in \cite[\S 5.3]{BronsteinBook} computes a triplet 
 $(g, p, s) \in K(t) \times K[t] \times S_t$ such that the above equation holds.

%%\vspace{-0.5cm}
%%\begin{center}
%\begin{alg}{\sc HermiteReduce} \label{ALG:hr}
%
%\noindent 
%%\fbox{
%%\parbox{0.46\textwidth}{
%{\sc Input:} $f \in K(t)$, where $t$ is a primitive monomial over $K$
%
%\noindent
%{\sc Output:} $(g,  p,  s) \in K(t) \times K[t] \times S_t$ such that \eqref{EQ:hr} holds
%\begin{enumerate}
%\item[1.] compute $(g, p_1, s_1) \in K(t) \times K[t] \times K(t)$
%by  such that $f = g^\prime + p_1 +  s_1$ and that $s_1$ has a normal
%denominator
%
%\item[2.] $p_2 \leftarrow\polypart(s_1)$,  $s_2 \leftarrow\proppart(s_1)$, {\sc return} $(g, \, p_1+ p_2, \, s_2)$
%\end{enumerate}
%%}}
%\end{alg}
%%\end{center}
%The algorithm is correct by Algorithm {\sc HermiteReduce}.
%\end{figure}
\begin{example} \label{EX:hr}
Let $K=C(x)$, $t = \log(x)$ and 
$$f= \frac{(x+1)t^2+(x^2+2x+2)t+x+1}{x(t+1)} \in K(t).$$
 {\sc HermiteReduce}$(f)$ finds $(g, p,  s) \in K(t) \times K[t] \times S_t$
such that \eqref{EQ:hr} holds, where
$g = 0,$  $p = \frac{x+1}{x} t+ \frac{{x}^{2}+x+1}{x},$ and $s = -\frac{x}{t+1}.$
Unfortunately,
the algorithm does not extract any in-field integrable part from $f$.
It will be shown that $p \in K(t)^\prime$ in Example \ref{EX:cr}.
\end{example}

The next lemma presents two properties concerning decomposition and contraction in primitive monomial extensions. 
 They play an important role in the proof of our main result (Theorem \ref{TH:decomp}).
\begin{lemma} \label{LM:reduce}
If $t$ is a primitive monomial over $K$, then 
\begin{itemize}
\item[(i)] $K(t)=(K(t)'+K[t])\oplus S_t$, and
\item[(ii)] $K(t)^\prime \cap K[t]=K[t]'.$
\end{itemize}
\end{lemma}
\begin{proof}
(i) holds by \eqref{EQ:hr}, and (ii)  holds because  the derivative of a proper element of $K(t)$ remains proper.
\end{proof}

\subsection{Effective bases} \label{SUBSECT:basis}

This section is a preparation for a dual technique to be used in  Sections \ref{SECT:basic} and \ref{SECT:tower}.
%Let $E$ be a field with a subfield $F$, and $\Theta$ be an $F$-linear basis  of $E$.
%For $\theta \in \Theta$,   $\theta^*$ stands for the $F$-linear function that maps $\theta$ to $1$ and any other element of $\Theta$ to $0$.
\begin{define} \label{DEF:effective}
Let $E$ be a field with a subfield $F$, $\Theta$ be an $F$-linear basis  of $E$, $\theta \in \Theta$ and $a \in E$. Then
\begin{itemize}
\item[(i)] $\theta^*$ stands for the $F$-linear function  on $E$ that maps $\theta$ to $1$ and any other element of $\Theta$ to $0$. 
\item[(ii)] $\theta$ is said to be {\em effective} for $a$ if $\theta^*(a) \neq 0$.
\item[(iii)] $\Theta$ is called an {\em effective $F$-basis} if there are two algorithms :
\begin{itemize}
\item one finds $\theta \in \Theta$ effective for $a$ if $a \neq 0$; and 
\item the other computes $\theta^*(a)$.
\end{itemize}
\end{itemize}
\end{define}

Let $F$ be a field and $F(y)$ the field of rational functions in $y$.
Set $Y = \left\{ y^i \mid i \in \bN \right\}$ and
$Q$ to be the set consisting of monic and irreducible polynomials with positive degrees.
Then
\begin{equation} \label{EQ:basis0}
  \Theta =  Y \cup \left\{ \frac{y^i}{q^j} \mid q \in Q, 0 \le i < \deg(q), j \in \bZ^+ \right\}
\end{equation}
is an effective $F$-basis of $F(y)$ by the irreducible partial fraction decomposition.
The two algorithms required in Definition \ref{DEF:effective} (iii) are given below.
Their correctness is evident.
\begin{alg}{\sc BasisElement} \label{ALG:be}

 \noindent 
%\fbox{
%\parbox{0.46\textwidth}{
{\sc Input}: $a \in F(y)^\times$ \quad {\sc Output}: $(\theta, c) \in  \Theta \times C^\times$ with $c = \theta^*(a)$
\begin{itemize}
	\item[1.] $p \leftarrow \polypart(a),$ $r \leftarrow \proppart(a)$, $d \leftarrow$ the denominator of $r$
	\item[2.] {\sc if} $p \neq 0$ {\sc then} {\sc return} $\left(y^{\deg(p)}, \, \lc(p) \right)$ {\sc end if} 
    \item[3.]  $q \leftarrow$ a factor of $d$ in $Q$, $m \leftarrow$ the multiplicity of $q$ in $d$
    \item[4.]  $h \leftarrow$ the coefficient of $q^{-m}$ in the $q$-adic expansion of $r$
    \item[5.]  {\sc return} $\left( y^{\deg(h)}/q^m,  \, \lc(h) \right)$
\end{itemize}
\end{alg}
\begin{remark} \label{RE:factor}
There is no obvious rule for choosing an irreducible factor $q$ of $d$ in step 3 of Algorithm \ref{ALG:be}. 
For example, let $f = \frac{1}{y(y+1)}$. One may set $q$ to be either $y$ or $y+1$. Then $\theta$ obtained in step 5
may be either $\frac{1}{y}$ or $\frac{1}{y+1}$. So the algorithm does not guarantee that
the same output will be returned when it is applied  to the same input twice. 
%However, this does not 
%For example, let $f = \frac{1}{x}-\frac{1}{x+1} \in \bQ(x)$.
\end{remark}
In practice, we choose $q$ to be the first member in the list of irreducible factors of $d$ computed by a  factorization algorithm.

\begin{alg} \label{ALG:cf} {\sc Coefficient}

 \noindent  
%\fbox{
%\parbox{0.46\textwidth}{
%{\rm
{\sc Input}: $(b, \theta) \in F(y) \times \Theta$ \quad {\sc Output}: $\theta^*(b)$
		
		\begin{itemize}
			\item[1.] $p \leftarrow \polypart(b),$ $r \leftarrow \proppart(b)$
            \item[2.] Write $\theta = y^k/q^m$ for some $k, m \in \bN$, $q \in Q$, $\gcd(y, q)=1$
			\item[3.] {\sc if} $m = 0$ {\sc then}
			{\sc return} the coefficient of $y^k$ in $p$
			{\sc end if}
			%\item[(4)] $q \leftarrow$polynomial in $Q$ dividing $d$, 
%$m \leftarrow$multiplicity of $q$ in $d$
			\item[4.] $h \leftarrow$ the coefficient of $q^{-m}$ in the $q$-adic expansion of $r$
			\item[5.] {\sc return} the  coefficient of $y^k$ in $h$
		\end{itemize}
%}
%}}
\end{alg}
%\end{figure}
 
\begin{remark} \label{RE:basis}
Let $F$ and $E$ be given in Definition \ref{DEF:effective} and $C$ a subfield of $F$. % be a field with a subfield $C$ and $E$ be given above.
Assume that $F$ has an effective $C$-basis $\Theta_0$ and that $E$ has an effective $F$-basis $\Theta$.
Then $\{\theta_0 \theta \mid \theta_0 \in \Theta_0, \theta \in \Theta\}$ is an effective $C$-basis of $E$
by a straightforward recursive argument.
\end{remark}

\subsection{Constant residues} \label{SUBSECT:log}

Let $(K, \, ^\prime)$ be a differential field with constant subfield $C$, and $t$ be a monomial over $K$.
For $f \in S_t$ and $\alpha \in \overline{K}$, an element 
$\beta \in \overline{K}$ is the {\em residue} of $f$ at $\alpha$
if and only if
$f=g+ \beta \frac{(t-\alpha)^\prime}{t-\alpha}$
for some $g \in \overline{K}(t)$ whose denominator is coprime with $t-\alpha$.
The residue of $f$ at $\alpha$ is nonzero if and only if $\alpha$ is a root of its denominator.
%%The basic properties of residues  are described in \cite[\S 4.4]{BronsteinBook}.
%Residues are closely related to logarithmic parts of elementary integrals (see \cite[\S 5.6]{BronsteinBook}, \cite[Theorem 3]{Raab2012}
%and \cite[\S 3]{DGGL2023}).
%For brevity, we say a residue of an element $s\in S$ means a residue of some nonzero component of $s$.

Below is a minor variant
of an algorithm described in the proof of \cite[Theorem 3.9]{RaabThesis}.
In its pseudo-code, $D_t$ stands for the derivation on $K(t)$ that maps every element of $K$ to $0$
and $t$ to $1$, and $\kappa$ for the coefficient-lifting derivation from $(K, \, ^\prime)$ to $K(t)$ (see \cite[\S 3.2]{BronsteinBook}).
\begin{alg}{\sc ConstantMatrix} \label{ALG:cs}

 \noindent 
{\sc Input:}  $f,g_1,\cdots,g_l \in S_t$

 \noindent
{\sc Output:} $M \in C^{k \times l} $  and $\vv \in C^k$ such that all residues of $f - \sum_{i=1}^{l}c_i g_i $
belong to $\overline{C}$  if and only if 
%\begin{equation} \label{EQ:constant}
$M \begin{pmatrix}
     c_1,
     \cdots,
     c_l
     \end{pmatrix}^\tau
     = \vv$ 
%\end{equation}
 %where $( \cdots )^\tau$ denotes the transpose of a matrix.

\begin{enumerate}
\item[1.] $h \leftarrow f - c_1 g_1 - \cdots - c_l g_l$, 

where $c_1, \ldots, c_l$ are constant indeterminates
\item[2.] $p \leftarrow$ the numerator of $h$, $q \leftarrow$ the denominator of $h$
\item[3.] $(u,v) \leftarrow$ the respective inverses of $(q^\prime, D_t(q))$ mod $q$ %, \, $v \leftarrow$ the inverse of $D_t(q)$ mod $q$

          $w \leftarrow \kappa( pu) - D_t(p u) \cdot v \cdot \kappa(q)$ 
          
          $r \leftarrow$ remainder of $w$ on division by $q$
\item[4.] $(M, \vv) \leftarrow$ an augmented matrix of the linear system 

in $c_1, \ldots, c_l$ obtained by
setting $r=0$
\item[5.] {\sc return} $M, \vv$
\end{enumerate}
%}
%}}
\end{alg}

To see its correctness, we note that $q$ obtained from step 2 is normal and free of $c_1, \ldots c_l$. Then $\gcd(q^\prime, q)=\gcd(D_t(q), q)=1$.
Therefore, both $u$ and $v$ can be computed in step 3.
Let $\alpha$ be a root of~$q$. Then $\alpha^\prime = - v(\alpha) \cdot \kappa(q)(\alpha)$ by \cite[Theorem 3.2.3]{BronsteinBook}.
On the other hand, the residue $\beta$ of $h$ at $\alpha$ is equal to $(pu)(\alpha)$ by \cite[Lemma 4.4.2]{BronsteinBook}.
 Then $\beta^\prime = w(\alpha)$, where $w$ is also
computed in step~3.
 Therefore, $r=0$ if and only if all  residues of $h$ belong to $\overline{C}$.
The system obtained in step 4 is linear
because $c_1, \ldots, c_l$ appear linearly in the coefficients of $r$.

\section{Basic constructions}\label{SECT:basic}

In this section, we let $(K, \, ^\prime)$ be a differential field and $C$ be the subfield of its constants.
Assume that there exists a complete reduction $\phi$ on $K$ for $K^\prime$, and
an algorithm that, for every $f \in K$, computes $g \in K$ such that $f = g^\prime + \phi(f)$. 
We call $\phi(f)$ the {\em remainder} of $f$ and $(g, \phi(f))$ a {\em reduction pair} of $f$ (with respect to~$\phi$). 
A reduction pair will be abbreviated as an R-pair in the sequel. 
 
%\begin{remark} \label{RE:lc}
%\begin{itemize} 
%\item[(i)] 
%$a - \phi(a) \in K^\prime$ for all $a \in K$. %by the definition of $\phi$  %In particular, $\lambda_t^\prime = t^\prime - \phi(t^\prime)$
%%for some fixed $\lambda_t \in K$. 
%\item[(ii)] $\phi(t^\prime) \in (K[t]^\prime \cap K)^\times$ because $t^\prime \in K \setminus K^\prime$. 
%\end{itemize}
%\end{remark}
Let $t$ be a primitive monomial over $K$.
We are going to define a complete reduction $\psi$ on $K(t)$
for $K(t)'$.
It suffices to construct a complementary subspace of $K[t]^\prime$ in $K[t]$ by Lemma \ref{LM:reduce}.
%In what follows, we work out the three steps outlined in \S \ref{SECT:intro}  with $W=K[t]$.

As a matter of notation, 
the $C$-subspace
$\bigoplus_{i \in \bN} V \cdot t^i$ for some $C$-subspace $V$ of $K$  is denoted by $V \otimes C[t]$
in virtue of the $C$-isomorphism $v \otimes t^i \mapsto v t^i$ from $V \otimes_C C[t]$ to $\bigoplus_{i \in \bN} V \cdot t^i$.

First, we decompose $K[t]$ as the sum of $K[t]^\prime$ and $\im(\phi) \otimes C[t]$. 
\begin{lemma} \label{LM:aux}
Let $p \in K[t]$ with $\deg(p)=d$. There exists $q \in K[t]$ with $\deg(q) \le d$ and $r \in \im(\phi) \otimes C[t]$ with $\deg(r) \le d$
such that
\begin{equation} \label{EQ:ar}
p = q^\prime + r.
\end{equation}
\end{lemma}
\begin{proof} If $p=0$, then set $q=r=0$. Assume that $p$ is nonzero with degree $d$ and leading coefficient $l$.

Let $\left(g, \phi(l) \right)$ be an R-pair of $l$, and $h=p-lt^d$. With integration by parts,  we have %that
\begin{equation} \label{EQ:normal}
	p  = g't^d + \phi(l)t^d +h
	  = \left(gt^d \right)' + \phi(l)t^d +  h- (d gt')t^{d-1}.
\end{equation}
Since $\phi(l) t^d \in \im(\phi) \otimes C[t]$ and $d> \deg\left(h - (d gt')t^{d-1}\right)$, 
the lemma follows from an induction on $d$.
\end{proof}
\begin{define} \label{DEF:aux}
The $C$-subspace $\im(\phi) \otimes C[t]$, denoted by $A$, is called the {\em auxiliary subspace} for $K[t]'$ in $K[t]$.
\end{define}
\begin{cor}\label{COR:aux}
$K[t] = K[t]^\prime + A.$
\end{cor}
\begin{proof} It is immediate from Lemma \ref{LM:aux}. \end{proof}
The next algorithm is direct from the proof of Lemma \ref{LM:aux}.

\begin{alg}
{\sc AuxiliaryReduction} \label{ALG:ar}

\noindent 
{\sc Input:}   $p \in K[t]$
%\item[] 

\noindent
{\sc Output:}  $(q, r) \in K[t] \times A$ such that \eqref{EQ:ar} holds
%\end{itemize}
\begin{enumerate}
\item[1.] $\tilde{p} \leftarrow p$, $q \leftarrow 0$, $r \leftarrow 0$
\item[2.] {\sc while} $\tilde{p} \neq 0$ {\sc do}
\begin{itemize}
\item[]  $d \leftarrow \deg(\tilde{p})$, $l \leftarrow \lc(\tilde{p})$, compute an R-pair $(g, \phi(l))$ of $l$

$q \leftarrow q + g t^d, \,\, r \leftarrow r + \phi(l) t^d, \,\, \tilde{p} \leftarrow \tilde{p} - l t^d - (d g t^\prime) t^{d-1}
$
\end{itemize}
{\sc end do}
\item[3.] {\sc return} $(q, r)$
\end{enumerate}
%}}
%}
\end{alg}
%\end{center}
%\begin{example} \label{EX:aux}
%Let $K(t)$ be the monomial extension in Example \ref{EX:hr},
%and $\phi:K \rightarrow K$ be the complete reduction in Example \ref{EX:rational}.
%Then $S_x \otimes C[t]$ is the auxiliary subspace $A$. 
%We reduce $p$ in Example \ref{EX:hr}. Algorithm \ref{ALG:ar} yields $p=q^\prime + r$, where
%$q = xt+ \frac{x^2}{2}$ and $r = \frac{t+1}{x} \in A.$
%With $f$ and $s$ given in Example \ref{EX:hr}, we have $f = q^\prime + r + s.$
%Algorithm \ref{ALG:ar} extracts a nontrivial derivative $q^\prime$ from $f$.
%\end{example}

Next, let us construct a $C$-basis of $K[t]' \cap A$. To this end,
we fix an R-pair $(\lambda_t, \phi(t^\prime))$ of $t^\prime$ and call
it {\em the first pair associated to $K(t)$}. 
\begin{remark} \label{RE:lc}
The remainder $\phi(t^\prime) \in K[t]^\prime$, because it is equal to $(t - \lambda_t)^\prime$.
Moreover, $\phi(t^\prime) \neq 0$ because $t$ is a primitive monomial.  
\end{remark}

For all $i \in \bZ^+,$  we calculate
\begin{equation} \label{EQ:basis1}
    \phi(t')t^i = t't^i - \lambda_t^\prime t^i 
	%= \left(\dfrac{t^{k+1}}{k+1}\right)^\prime - \lambda_t't^k \\
   = \left(\dfrac{t^{i+1}}{i+1} - \lambda_tt^i \right)' + (i \lambda_t t') t^{i-1}.
\end{equation}
There exists a pair  $(q_i, r_i) \in K[t] \times A$ such that 
$(i \lambda_t t') t^{i-1} = q_i^\prime + r_i$ and $\deg(r_i) \le i-1$
by Lemma \ref{LM:aux}. It follows that
\begin{equation} \label{EQ:basis2}
\phi(t')t^i  - r_i = \left(\frac{t^{i+1}}{i+1} - \lambda_t t^i  + q_i \right)^\prime.
\end{equation}
%Then $v_i \in K[t]^\prime$. Since both $\phi(t')t^i$ and $r_i$ belong to $A$, we have that $v_i \in A$. 
%So $v_i \in K[t]' \cap A$. Set $v_0 = \phi(t^\prime)$, which also belongs to $K[t]' \cap A$ by Remark~\ref{RE:lc}. 
%In addition, $\deg(v_i)=i$ for all $i \in \bN$.
%Set $u_0 := (t -\lambda_t)^\prime$ and, for all $k \in \bZ^+$,
%\[ u_k := \dfrac{t^{k+1}}{k+1} - \lambda_t t^k - p_k.\]
%Then $u_k^\prime = v_k$ for all $k \in \bN$.
\begin{lemma} \label{LM:basis}
Let $v_0 = \phi(t^\prime)$ and $v_i$ be the left-hand side of \eqref{EQ:basis2}. Then
\begin{itemize}
\item[(i)] $\deg(v_i)=i$ and $\lc(v_i)=\phi(t^\prime)$ for all $i \in \bN$.
\item[(ii)] 
The set $\{ v_0, v_1, \ldots \}$ is 
a $C$-basis of $K[t]' \cap A$.
\end{itemize}
%where $v_0 = \phi(t^\prime)$ and $v_k$ is given in \eqref{EQ:basis2} for $k \in \bZ^+$.}
\end{lemma}
\begin{proof}
(i) holds because $\phi(t^\prime) \neq 0$ and $r_i$ in \eqref{EQ:basis2} has degree $<i$.

(ii) Set $I=K[t]^\prime \cap A$. Then $v_0 \in I^\times$
 by Remark \ref{RE:lc} and Definition \ref{DEF:aux}. 
For $i>0$, $v_i \in K[t]^\prime$ by \eqref{EQ:basis2}. It is in $A$ because $\phi(t^\prime) t^i, r_i \in A$.
Thus, $v_i \in I$ for all $i \in \bN$.
The $v_i$'s are $C$-linearly independent by (i).

Assume that $p \in I$. 
Then $p \in K(t)' \cap K[t]$. It follows from  \cite[Lemma 2.3]{CDL2018} that
$\lc(p) = c t^\prime + b^\prime$ for some $c \in C$ and $b \in K$.
%Then $p \in K[t]^\prime$.
%It follows from \cite[Lemma 2.3]{CDL2018} and Lemma \ref{LM:reduce} that
%$\lc(p) = c t^\prime + b^\prime$ for some $c \in C$ and $b \in K$.
On the other hand, $p \in A$ implies that $\lc(p) \in \im(\phi)$.
Hence, applying $\phi$ to $\lc(p) = c t^\prime + b^\prime$ yields
$\lc(p)=c \phi(t^\prime),$ because $\phi$ is
an idempotent and  $\phi(b')=0$.
Let $i=\deg(p)$ and $q = p - c v_i$. Then $q \in I$ with $\deg(q)<i$.
Thus, $p$ is a $C$-linear combination of $v_0, \ldots, v_i$ by a straightforward induction on $i$. 
\end{proof}
\begin{remark} \label{RE:risch}
\cite[Lemma 2.3]{CDL2018} cited in the above proof combines two results in \cite{BronsteinBook} into one statement. 
A referee of our manuscript points out that the two results are given in equation (5.13)
and the second paragraph on
page 176, respectively.
\end{remark}

Calculations in \eqref{EQ:basis1} and  \eqref{EQ:basis2} lead to a naive 
algorithm to construct the basis $\{v_0, v_1, \ldots \}$ of $K[t]^\prime \cap A$ in Lemma \ref{LM:basis} (ii) up to a given degree $d$.
A more elaborated algorithm suggested  by a referee minimizes R-pairs to be computed. To present the algorithm, we need some notation
and a technical lemma.

Let $(\lambda_t, \phi(t^\prime))$ be the first pair associated to $K(t)$. Set $\mu_0 = \lambda_t$ and
$\left(\mu_{k+1}, \nu_{k+1}\right)$ to be an R-pair of $\mu_k t^\prime$ with respect to $\phi$ for all $k \in \bN$.
 Furthermore, 
 we define two families of differential operators:
\begin{equation} \label{EQ:ops}
  L_{i,j} = \sum_{k=1}^i (-1)^{k+1} \mu_{j+k} D_t^{k} \quad \text{and} \quad
    M_{i,j} = \sum_{k=1}^i (-1)^{k+1}  \nu_{j+k}  D_t^{k} 
\end{equation}
for all $i \in \bZ^+$ and $j \in \bN$, where $D_t$ is the same as in Section \ref{SUBSECT:log}.
 \begin{lemma} \label{LM:rec}
With the notation just introduced, we have 
$$i \mu_j t^\prime t^{i-1} =  \left(L_{i, j}(t^i)\right)^\prime  + M_{i,j}(t^i)$$
for all $i \in \bZ^+$ and $j  \in \bN$. 
\end{lemma}
\begin{proof}
We proceed by induction on $i$ and regard $j$ as an arbitrary nonnegative integer. 
For $i = 1$,  $\mu_j t^\prime  = \mu_{j+1}^\prime  + \nu_{j+1}$ by definition.  So $\mu_j t^\prime  =  L_{1, j}(t)^\prime + M_{1,j}(t)$  by \eqref{EQ:ops}.
The conclusion holds for $i=1$. 

Assume that $i>1$ and the conclusion holds for values lower than $i$ and every $j \in \bN$. We calculate
\begin{align*}
 \mu_j t^\prime t^{i-1} = & \,  \left(\mu_{j+1}^\prime +  \nu_{j+1} \right) t^{i-1} \\
=  & \,  \left(\mu_{j+1} t^{i-1} \right)^\prime -  \mu_{j+1} \left( t^{i-1}\right)^\prime +  \nu_{j+1} t^{i-1}  \\
= & \,   \left(\mu_{j+1} t^{i-1} \right)^\prime -  (i-1) \mu_{j+1} t^\prime t^{i-2} + \nu_{j+1} t^{i-1}  \\
= & \left( \mu_{j+1} t^{i-1} {-}  L_{i-1,j+1}(t^{i-1}) \right)^\prime  
{+}  \nu_{j+1} t^{i-1} {-}  M_{i-1,j+1}(t^{i-1}) \\
= & i^{-1} \left(L_{i, j}(t^i)\right)^\prime  +  i^{-1} M_{i,j}(t^i),
\end{align*}
in which the first equality holds because $ \mu_j t^\prime = \mu_{j+1}^\prime + \nu_{j+1};$ 
 the second is derived from integration by parts, the third is by a direct calculation; the fourth is due to the induction hypothesis 
$$(i-1) \mu_{j+1} t^\prime t^{i-2} = \left(L_{i-1,j+1}(t^{i-1})\right)^\prime  + M_{i-1,j+1}(t^{i-1});$$
and the last holds by replacing $i D_t^k(t^{i-1})$ with $D_t^{k}(t^i)$ and shifting the index $k$ to $k+1$ in \eqref{EQ:ops}. 
\end{proof}

\begin{cor} \label{COR:basis}
For all $i \in \bZ^+$, the element $v_i$ in Lemma \ref{LM:basis} can be taken
as $\phi(t^\prime) t^i - M_{i,0}(t^i)$, which is equal to 
$\left( \frac{t^{i+1}}{i+1} - \lambda_t t^i + L_{i,0}(t^i) \right)^\prime.$
\end{cor}
\begin{proof} It follows from Remark \ref{RE:lc} that 
\[ \phi(t^\prime) t^i = \left( \frac{t^{i+1}}{i+1} - \lambda_t t^i \right)^\prime + i \lambda_t t^\prime t^{i-1}. \]
Setting $j=0$ in Lemma \ref{LM:rec} and noticing $\mu_0 = \lambda_t$, we see that 
\[ \phi(t^\prime) t^i = \left( \frac{t^{i+1}}{i+1} - \lambda_t t^i + L_{i,0}(t^i) \right)^\prime + M_{i,0}(t^i). \]
Since $M_{i,0}(t^i)$ belongs to $A$ and is of degree less than $i$, the corollary holds when we set $q_i = L_{i,0}(t^i)$
and $r_i = M_{i,0}(t^i)$ in \eqref{EQ:basis2}. 
\end{proof}

%The next algorithm constructs the basis in the above lemma up to a given degree.
%It is correct by \eqref{EQ:basis1} and  \eqref{EQ:basis2}.

\begin{alg}{\sc Basis} \label{ALG:basis}

\noindent  
{\sc Input:} $d \in \bN$ and the first pair $(\lambda_t, \phi(t^\prime))$ associated to $K(t)$ 

\noindent
{\sc Output:}  
 $(u_0,v_0), \ldots, (u_d,v_d)$, where  
$(u_0, v_0) {=} \left(t - \lambda_t, \, \phi(t^\prime) \right)$
and, for  all  $i \in [d]$,
$
(u_i, v_i) = \left(\frac{t^{i+1}}{i+1} - \lambda_t t^i + L_{i,0}(t^i), \, \phi(t^\prime) t^i - M_{i,0}(t^i)\right)$
with $L_{i,0}$ and $M_{i,0}$ given in \eqref{EQ:ops}.

%\end{itemize}
%\begin{enumerate}
%%\item[1.] compute $\lambda_t \in K$ such that $\lambda_t^\prime = t^\prime - \phi(t^\prime)$
%\item[1.] $L \leftarrow [(t - \lambda_t, \phi(t^\prime))]$
%\item[2.] {\sc for} $i$ {\sc from} $1$ {\sc to} $d$ {\sc do}
%\begin{itemize}
%\item[] $a \leftarrow t^{i+1}/(i+1) - \lambda_t t^i$, \,\, $b \leftarrow (i\lambda_t t')t^{i-1}$
%\item[] $(q, r)\leftarrow$ {\sc AuxiliaryReduction}$(b)$ \,\,  
%$(^*$Algorithm \ref{ALG:ar}$^*)$
%\item[] $(u, v) \leftarrow (a + q, \, \phi(t^\prime) t^i - r)$
%\item[] $L \leftarrow$  the list obtained by appending $(u,v)$ to $L$
%\end{itemize}
%{\sc end do}
%\item[3.] {\sc return} $L$
%\end{enumerate}
%%}}
%%}
\end{alg}
We skip the pseudo-code for the above algorithm, because both $L_{i,0}(t^i)$ and $M_{i,0}(t^i)$ can be easily
constructed iteratively according to \eqref{EQ:ops}. In the algorithm, one  computes R-pairs
$(\mu_i, \phi(\mu_{i-1} t^\prime))$ for  every $i \in [d]$, whereas the number of R-pairs needed is proposional
to $d^2$ if one  constructs $v_1, \ldots, v_d$ by Algorithm \ref{ALG:ar} directly.

%\end{center}
%\begin{example} \label{EX:basis}
%For  the monomial extension given in Example \ref{EX:hr}, the first three elements in the distinct-degree basis of $I$
%are:
%The above algorithm also computes
%$u_0 = t$, $u_1=t^2/2$ and $u_2 = t^3/3$ such that $v_i=u_i^\prime$, $i=0,1,2$.
%\end{example}

Now, we turn the sum in Corollary \ref{COR:aux} to a direct one by constructing a subspace of $A$ that is a complement of $K[t]^\prime$.
To proceed,  
we need to assume further that $K$ has an effective  $C$-basis, which is denoted by $\Theta$. 
Then there exists a pair $(\theta, c) \in \Theta \times C^\times$ such that $c = \theta^*(\phi(t^\prime))$. 
We fix such a pair and call it the {\em second pair associated to $K(t)$}. A complementary subspace
consists of polynomials in $A$ whose coefficients are free of $\theta$. 
In other words, the subspace is equal to $\left(\im(\phi) \cap \ker(\theta^*) \right)\otimes C[t]$.

\begin{lemma} \label{LM:direct}
 Let $(\theta,c)$ be the second pair associated to $K(t)$. Then
 \begin{itemize}
 \item[(i)] $A = ( K[t]' \cap A )\oplus A_\theta,$
where  
$A_\theta = \left(\im(\phi) \cap \ker\left(\theta^* \right) \right) \otimes C[t];$ 
\item[(ii)] $K[t]=K[t]^\prime \oplus A_\theta$.
\end{itemize}
 \end{lemma}
 \begin{proof} 
 (i) Similar to the proof of Lemma \ref{LM:basis}, we set $I= K[t]' \cap A $.
 
 First, we show  $A = I +  A_\theta.$
 Since $I \subset A$ and $A_\theta \subset A$,
 it suffices to show  $A \subset  I +  A_\theta.$
 Let $\{ v_0, v_1,  \ldots \}$ be the basis of $I$ in Lemma \ref{LM:basis} (ii), and $p \in A$.
 Set $d {=} \deg(p)$, $l {=} \lc(p)$ and $z {=} \theta^*(l)$. By Lemma \ref{LM:basis} (i),   
 \begin{equation} \label{EQ:elim}
 p   - c^{-1} z v_d =  g t^d + h,
 \end{equation}
 where $g=l - c^{-1} z \phi(t^\prime)$ and $h \in K[t]$ with $\deg(h)<d$.
 Since $p \in A$, we have that $ l \in \im(\phi)$, and, thus, $g \in \im(\phi)$ by its definition.
 Furthermore, $\theta^*(g) = \theta^*(l) - c^{-1} z \theta^*(\phi(t^\prime)) = z - z=0$.
 Hence, $g \in \ker(\theta^*)$. Consequently, $g \in \im(\phi) \cap \ker(\theta^*)$. We conclude that $g t^d \in A_\theta$.
 It follows from \eqref{EQ:elim} that $h \in A$ and $p - h \in I + A_\theta$, which allow us to
 carry out an induction on $d$ as follows.
 
 If $d=0$, then $h=0$. So $p \in I +  A_\theta.$ 
 Assume that $d>0$ and that all elements of~$A$ with degree lower than $d$ are in $I +  A_\theta.$ 
 Then $h \in I +  A_\theta.$  Hence, $p \in I +  A_\theta$.

 % with $d=\deg(p)$, and let $v_d$ be the element of $B$ with $\deg(v_d)=d$.
% Then there exists $r \in \im(\phi) \cap \ker(\theta^*)$ and $s \in A$ with $\deg(s)<d$ such that 
% $
% p   - \theta^*(\lc(p)) c^{-1} v_d = r t^d + s.
% $
% So $r t^d \in A_\theta$.
% Repeating the same argument on $s$, we see that $A \subset (K[t]' \cap A ) + A_\theta$.
% 

 Second, we show that $I \cap A_\theta = \{0\}$.
 Assume that $q \in I \cap  A_\theta$. Then $q$ is a $C$-linear combination of the $v_i$'s. 
 So  
 $\lc(q)$ is the product of a constant and $\phi(t^\prime)$ by  Lemma \ref{LM:basis} (i).
 Since $\lc(q) \in \ker(\theta^*)$ and $\phi(t^\prime) \notin \ker(\theta^*)$, the constant is equal to zero, and so is $\lc(q)$.
 Accordingly, $q=0$.

 (ii) By Corollary \ref{COR:aux} and (i), $K[t]=K[t]^\prime + A_\theta$.
Since $A_\theta \subset A$, we have that
 $K[t]^\prime \cap A_\theta = \left(K[t]^\prime \cap A \right)\cap A_\theta$, which is equal to $\{0\}$ by (i). 
So (ii) holds. 
 \end{proof}

The $C$-subspace  $A_\theta$ in Lemma \ref{LM:direct} is called the {\em $\theta$-complement} of $K[t]^\prime$ in $K[t]$ in the rest 
of this section.
%To develop an algorithm to project an element of $K[t]$ to $K[t]^\prime$ and the $\theta$-complement, respectively,
%we assume further that the $C$-basis $\Theta$ in the above lemma is effective. Moreover, we fix a pair $(\theta, c) \in \Theta \times C^\times$
%such that $c=\theta^*(\phi(t^\prime))$, and call it the {\em second pair associated to $K(t)$}.
%The above proof leads to
The next algorithm projects an element of $A$ to $K[t]^\prime$ and the $\theta$-complement, respectively.
It is correct by the proof of Lemma \ref{LM:direct} (i).
\begin{alg}{\sc Projection} \label{ALG:proj}

\noindent 
{\sc Input:} $r \in A$,  the first and second pairs $(\lambda_t, \phi(t^\prime))$ and $(\theta, c)$  
 associated to $K(t)$ 

\noindent 
{\sc Output:} $(u, v) \in K[t] \times A_\theta$ such that
\begin{equation} \label{EQ:proj}
r = u^\prime + v
\end{equation}
%\end{itemize}
\begin{enumerate}
\item[1.] $u \leftarrow 0$, $v \leftarrow r$, $d \leftarrow \deg(r)$
\item[2.] $B \leftarrow$ {\sc  Basis}$(d, \lambda_t, \phi(t^\prime))$ \,\,  $(^*$Algorithm \ref{ALG:basis} $^*)$
\item[3.] {\sc for} $i$ {\sc from} $0$ {\sc to} $d$ {\sc do}
\begin{itemize}
\item[]$a \leftarrow$ the coefficient of $t^{d-i}$ in $v$, 
$b \leftarrow \theta^*(a)$

$(\tilde{u}, \tilde{v}) \leftarrow$ the element of $B$ with 
$\deg(\tilde v)=d-i$,

$\tilde{c} \leftarrow c^{-1}b,$ \, $u \leftarrow u + \tilde{c} \tilde{u}$, \, $v \leftarrow v - \tilde{c} \tilde{v}$
\end{itemize}
{\sc end do}
\item[4.] {\sc return} $(u, v)$
\end{enumerate}
%}}
%}
\end{alg}

We are ready to present the main result of this section.
\begin{thm} \label{TH:decomp}
 Let $(\theta,c)$ be the second pair associated to $K(t)$, and $A_\theta$ be the $\theta$-complement of $K[t]^\prime$.
Then
%\begin{equation} \label{EQ:decomp}
  $K(t) = K(t)^\prime \oplus A_\theta  \oplus S_t.$
%\end{equation}
Moreover, the projection $\psi_\theta$ from $K(t)$ to $A_\theta \oplus S_t$ with respect to the above
direct sum is a complete reduction for $K(t)^\prime$.
\end{thm}
\begin{proof} By Lemma \ref{LM:reduce} (i) and Lemma \ref{LM:direct}, $$K(t) = (K(t)^\prime  + A_\theta) \oplus S_t.$$ 
By Lemma \ref{LM:reduce} (ii) and $A_\theta \subset K[t]$, 
we have $K(t)^\prime \cap A_\theta = K[t]^\prime \cap A_\theta$,
which 
is trivial by Lemma \ref{LM:direct} (ii). So $K(t) = K(t)^\prime \oplus A_\theta  \oplus S_t$.
It follows that  $\psi_\theta$ is a complete reduction for $K(t)^\prime$. 
\end{proof}

Below is an algorithm for the complete reduction given in the above theorem.

\begin{alg}{\sc CompleteReduction} \label{ALG:cr}

\noindent  
{\sc Input:} $f \in K(t)$, the first and second pairs $(\lambda_t, \phi(t^\prime))$ and $(\theta, c)$  
 associated to $K(t)$ 

\noindent  
{\sc Output:} an $R$-pair of $f$ with respect to $\psi_\theta$ in Theorem \ref{TH:decomp}
\begin{enumerate}
\item[1.] $(g, p, s) \leftarrow$ {\sc HermiteReduce}$(f)$ \,\,  
$(^*$\cite[\S 5.3]{BronsteinBook} $^*)$

{\sc if} $p=0$ {\sc then} {\sc return} $(g, s)$ {\sc end if}
\item[2.] $(q, r) \leftarrow$ {\sc AuxiliaryReduction}$(p)$ \,\, $(^*$Algorithm \ref{ALG:ar} $^*)$

{\sc if} $r=0$ {\sc then} {\sc return} $(g+q, s)$ {\sc end if}

%\item[3.] Find $(\theta,c) \in \Theta \times C^\times$ with $\theta$ effective for $\phi(t^\prime)$ and 

%$c = \theta^*\left(\phi(t^\prime)\right)$ 
%by Algorithm \ref{ALG:be}
\item[3.] $(u,v) \leftarrow$ {\sc Projection}$(r, \lambda_t, \phi(t^\prime), \theta, c)$ \,\, $(^*$Algorithm \ref{ALG:proj} $^*)$

{\sc return} $(g+q+u, s+v)$
\end{enumerate}
%}}
%}
\end{alg}
%\end{table}
%\end{center}
%The correctness of this algorithm follows immediately from \eqref{EQ:hr}, \eqref{EQ:ar}, \eqref{EQ:proj} and Theorem \ref{TH:decomp}.
\begin{example} \label{EX:cr}
Let $K(t)$ and $f$ be given in Example \ref{EX:hr}, and $\Theta$ be the $C$-basis given in \eqref{EQ:basis0} with $F=C$ and $y=x$.
The first and second  associated pairs are $(0, x^{-1})$ and $(x^{-1}, 1)$,
respectively. 
The above algorithm computes an $R$-pair of $f$ as follows. 

\begin{enumerate}
\item[1.] $(g,p,s)=\left(0, \frac{x+1}{x} t+ \frac{{x}^{2}+x+1}{x},-\frac{x}{t+1}\right)$ 
%such that \eqref{EQ:hr} holds 
by Example \ref{EX:hr}.
\item[2.] Algorithm \ref{ALG:ar} finds $(q,r)=\left(xt+ \frac{x^2}{2}, \frac{t+1}{x}\right) \in K[t] \times A$ such that \eqref{EQ:ar} holds, 
where $A = S_x \otimes C[t]$. 
%\item[3.] Since $\phi(t^\prime)=x^{-1}$, we find $\theta=x^{-1}$ with $\theta^*(\phi(t^\prime))=1$.
\item[3.] Algorithm \ref{ALG:proj} finds $(u,v) = \left(\frac{t^2}{2}+t, 0\right)$ such that \eqref{EQ:proj} holds.
\end{enumerate}
Thus, $p=(q+u)^\prime$ and $(g+q+u, s)$ is an R-pair of $f$. 
Algorithm \ref{ALG:cr} finds $s = -\frac{x}{t+1}$ as a \lq\lq minimal\rq\rq\ non-in-field integrable part.
\end{example}

%In step 3 of Algorithm \ref{ALG:cr}, we choose $\theta \in \Theta$ only once.
%The ambiguity mentioned in Remark \ref{RE:factor} will not occur. 
%It happens quite often in reduction-based creative telescoping that we apply 
%a complete reduction to several elements of $K(t)$. 
%So we need to avoid reducing them to different complementary subspaces of $K(t)^\prime$.
%This difficulty  will be overcome in Section \ref{SECT:tower}.

At last, we describe the restriction of $\psi_\theta$ to $K$.
\begin{cor} \label{COR:cr}
Let $\phi: K \rightarrow K$ be a complete reduction for $K^\prime$, $(\theta,c)$ be the second pair associated to $K(t)$
and $\psi_\theta$ be the complete reduction given in Theorem \ref{TH:decomp}. Then,
for every $f \in K$, we have that $\psi_\theta(f) = \phi(f) + \tilde{c} \phi(t^\prime)$, where $\tilde{c}  = - \theta^*\left(\phi(f)\right) c^{-1}.$
\end{cor}
\begin{proof} Since $f \in K$, we have $f \equiv \phi(f)$ mod $K^\prime$.
By Remark \ref{RE:lc},  $f \equiv \phi(f)+ \tilde{c} \phi(t^\prime)$ mod $K(t)^\prime$.
Note that $\phi(f) + \tilde{c} \phi(t^\prime)$ belongs to the $\theta$-complement.
Applying $\psi_\theta$ to the above congruence, we conclude that $\psi_\theta(f)=\phi(f)+ \tilde{c} \phi(t^\prime),$
because $K(t)^\prime = \ker(\psi_\theta)$ and the restriction of $\psi_\theta$ to $A_\theta$ is the identity map.   
%{\blue So $\psi_\theta(f) = \phi(f) + \tilde{c} \phi(t^\prime)$ by Theorem \ref{TH:decomp}}.
%Since $f \in K$, we have $f = g^\prime + \phi(f)$ for some $g \in K$.
%So
%$\psi_\theta(f) = \psi_\theta(g^\prime) + \psi_\theta \left( \phi(f) \right) = \psi_\theta \left( \phi(f) \right).$
%Note that $\phi(f)$ is in the auxiliary subspace $A$ and of degree $0$. 
%Then $\phi(f) + \tilde{c} \phi(t^\prime)$ belongs to the $\theta$-complement.
%{\blue So $\psi_\theta(f) = \phi(f) + \tilde{c} \phi(t^\prime)$ by Theorem \ref{TH:decomp}}.
\end{proof}

\section{Complete reduction} \label{SECT:tower}

In this section, we define primitive towers and remove the assumptions made in 
the first paragraph of Section \ref{SECT:basic}. %previous section. 
\begin{define} \label{DEF:tower}
Let $K_0$ be a differential field whose subfield of constants is denoted by $C$.
A {\em primitive tower} over $K_0$ is
\begin{equation} \label{EQ:tower}
 \begin{array}{ccccccc}
    K_0       & \subset &  K_1      & \subset & \cdots         & \subset     & K_n  \\
              &         & \shortparallel &         &            &             & \shortparallel \\
              &         &  K_0(t_1) &         &                   &  & K_{n-1}(t_n),
\end{array}
\end{equation}
where $t_i$ is a primitive monomial over $K_{i-1}$ for all $i \in [n]$.
\end{define}
Note that $C$ is the subfield of constants in a primitive tower $K_n$. 
\begin{thm} \label{TH:tower}
Let $K_n$  be a primitive tower as in \eqref{EQ:tower}, and $\Theta_0$ be an effective $C$-basis of $K_0$. 
Assume that $\phi_0:K_0 \rightarrow K_0$ is a complete reduction for $K_0^\prime$,
and that there is an algorithm to compute an $R$-pair of every element in $K_0$. 
Then, for every $ i \in [n]_0$, $K_i$ has an effective $C$-basis $\Theta_i$ and a complete reduction $\phi_i: K_i \rightarrow K_i$ for $K_i^\prime$. Moreover, 
there is an algorithm to compute an $R$-pair of every element in $K_i$.
%\item[(ii)] For every $i \in [n-1]_0$ and $f \in K_i$,
%$\phi_n(f)-\phi_i(f)$ is a $C$-linear combination of $\phi_i(t_{i+1}^\prime),$ \ldots, $\phi_{n-1}(t_n^\prime)$, and belongs to $K_{n}^\prime$.
%%In particular, $\phi_n(f)-\phi_i(f) \in K_{n}^\prime$. 
%\end{itemize}
\end{thm}
\begin{proof} We proceed by induction on $n$. If $n=0$, then the conclusion clearly holds.
 Assume that $n>0$ and that 
there exists an effective $C$-basis $\Theta_{n-1}$ of $K_{n-1}$, a complete reduction $\phi_{n-1}$  for $K_{n-1}^\prime$ on $K_{n-1}$
and an algorithm to compute an R-pair of every element in $K_{n-1}$. 
The first and second pairs $\left(\lambda_n, \phi_{n-1}(t_n^\prime) \right)$ and $(\theta_n, c_n)$ 
associated to $K_{n}$ can be constructed by $\phi_{n-1}$ and $\Theta_{n-1}$, respectively.

The tower $K_n$ has an effective $C$-basis $\Theta_n$ by Remark \ref{RE:basis}.
%Since $t_n^\prime \in K_{n-1}$ and {\blue $\phi(t_n^\prime) \neq 0$},  there is $\theta_n \in \Theta_{n-1}$ effective in $\phi_{n-1}(t_n^\prime)$.
Replacing $K$ with $K_{n-1}$, $t$ with $t_n$, $\phi$ with $\phi_{n-1}$, and $\theta$ with $\theta_n$
in Theorem \ref{TH:decomp},  we find a complete reduction
$\psi_{\theta_n}$ for $K_{n}^\prime$ on $K_n$.  
Doing the same replacements in
Algorithms \ref{ALG:ar}, \ref{ALG:basis}, \ref{ALG:proj} and \ref{ALG:cr}, we have
an algorithm to compute an R-pair of every element in $K_n$ with respect to $\psi_{\theta_n}$. 
The induction is completed by setting $\phi_n=\psi_{\theta_n}$.
\end{proof}
\begin{cor} \label{COR:tower}
Let $K_n$  be a primitive tower as in \eqref{EQ:tower} and $\phi_i$ be the complete reduction constructed 
in the proof of the above theorem. Then, for every $i \in [n-1]_0$ and $f \in K_i$,
$\phi_n(f)-\phi_i(f)$ is a $C$-linear combination of $\phi_i(t_{i+1}^\prime),$ \ldots, $\phi_{n-1}(t_n^\prime)$, and belongs to $K_{n}^\prime$.
%%In particular, $\phi_n(f)-\phi_i(f) \in K_{n}^\prime$. 
\end{cor}
\begin{proof} For every $j \in [n-1]_0$, $\phi_{j+1}(f)-\phi_j(f)=c_j \phi_{j}(t_{j+1}^\prime)$
for some $c_j \in C$ by Corollary \ref{COR:cr}. Summing up these equalities from $i$ to $n-1$,
we see that $\phi_n(f)-\phi_i(f)$ is a $C$-linear combination of $\phi_i(t_{i+1}^\prime),$ \ldots, $\phi_{n-1}(t_n^\prime)$.
It belongs to $K_{n}^\prime$ by Remark \ref{RE:lc}.
\end{proof}
%\end{proof}
%\in \{ithere are $c_i, \ldots, c_{n-1} \in C$ such that
%$$ \phi_{i+1}(f)-\phi_i(f)=c_i \phi_{i}(t_{i+1}^\prime),
%\ldots,
%\phi_{n}(f)-\phi_{n-1}(f)=c_{n-1} \phi_{n-1}(t_{n}^\prime).$$ 
%{\blue The conclusion holds by summing up these equalities.} 
%%Furthermore,  $\phi_n(f)-\phi_i(f) \in K_{n}^\prime$
%%by Remark \ref{RE:lc} (ii). 

%\begin{define} \label{DEF:rem}
%Let $K_i$ be given in \eqref{EQ:tower} and $\phi_i: K_i \rightarrow K_i$ be a complete reduction for $K_i^\prime$
%for all $i \in [n]_0$. Let $f \in K_i$. Then $\phi_i(f)$ is called the {\em remainder} of $f$ with respect to $\phi_i$,
%and $(g, \phi_i(f))$ a {\em Risch pair} of $f$ if $f = g^\prime + \phi_i(f)$ with $g \in K_i$.
%\end{define}

To perform complete reductions in practice,  we assume further that $[K_0: C(x)] < \infty$
and that $K_0$ contains no new constant.
Complete reductions on $C(x)$ and its finite algebraic extensions are
given in Example \ref{EX:rational} and \cite{CKK2016}, respectively.
Improvements on the reduction for algebraic functions can be found in \cite{CDK2021}. 
Algorithms \ref{ALG:be} and \ref{ALG:cf} show that  $C(x)$ has an effective $C$-basis. 
So does $K_0$ by Remark \ref{RE:basis}.  
 Consequently,  a complete reduction  for $K_n^\prime$ on $K_n$ is available by Theorem \ref{TH:tower}.

% 
%The first and second pairs associated to $K_i$ for every $i \in [n]$}.
%The first  is an R-pair ${\blue (\lambda_i, \phi_{i-1}(t_i^\prime))}$ of $t_i^\prime$, and
%the second is $(\theta_i, c_i)$, where $\theta_i \in \Theta$ is effective in $\phi_{i-1}(t_i^\prime)$ and 
%{\blue $c_i$ is} $\theta_i^* \left( \phi_{i-1}(t_i^\prime) \right)$.
%%The associated pairs make sure that the complete reduction
%%operators are the same when we compute Risch pairs of several elements.
%%They also avoid some repeated computation. 
%The two pairs can be computed by Algorithm \ref{ALG:cr} and \ref{ALG:be}, respectively. 
%As long as these pairs are available, step 1 in Algorithm \ref{ALG:basis} and step 3 in Algorithm \ref{ALG:cr}
%will be skipped. The former avoids computing an R-pair of $t_i^\prime$ repeatedly.
%The latter removes any ambiguity caused by choosing an element
%from $\Theta$.

Let us make a notational convention so that we can illustrate computations and proofs through a primitive tower concisely. 
\begin{convention} \label{CON:tower}
Let $K_n$  be a primitive tower as in \eqref{EQ:tower}, and $\phi_0$ be a complete reduction for $K_0^\prime$ on $K_0$.  
Let $\Theta$ be the effective $C$-basis of $K_n$ obtained from a repeated use of Remark \ref{RE:basis}. 
For all $i \in [n]$,
\begin{itemize}
\item 
$\phi_i: K_i \rightarrow K_i$ stands for the complete reduction
for $K_i^\prime$  in the proof of Theorem \ref{TH:tower},  
\item $\left(\lambda_i, \phi_{i-1}(t_i^\prime) \right)$ and $\left(\theta_i, c_i\right)$ for the first and second pairs associated to $K_i$, respectively,
\item $S_i $ for the set of simple elements in $K_i$ with respect to $t_i$, and 
\item $A_i$ for the auxiliary subspace in $K_{i-1}[t_i]$. 
\end{itemize}
\end{convention}
All associated pairs are constructed once and for all. So the possible ambiguity mentioned in Remark \ref{RE:factor} will never occur.
\begin{example} \label{EX:rattower}
Let $K_0 {=} C(x)$,
$t_1 {=} \log(1-x),$ and 
$t_2 {=} \polylog(2,x)$, which is equal to  $- \int \frac{\log(1-x)}{x}.$
Then $K_2=K_0(t_1,t_2)$ is a primitive tower. 
We associate  
$(\lambda_1, \phi_0(t_1^\prime)) = \left(0, \, \frac{1}{x-1} \right),$ $(\theta_1, c_1) = \left(\frac{1}{x-1}, \, 1 \right)$
and $(\lambda_2, \phi_1(t_2^\prime)) = \left(0, \,  -\frac{t_1}{x} \right)$, $(\theta_2, c_2) = \left(\frac{t_1}{x}, \,  -1 \right)$
to $K_1$ and $K_2$, respectively. 
%\begin{center}
%\begin{tabular}{|c|c|c|c|} \hline
%      $i$ &  $(\lambda_i, \phi_0(t_i^\prime))$ &   $(\theta_i, c_i)$ & $B_i$        \\ \hline
%    $1$ & $\left(0, \, \frac{1}{x-1} \right)$  & $\left(\frac{1}{x-1}, \, 1 \right)$ &  $\left\{ \left(t_1, \, \frac{1}{\blue x-1}  \right)\right\}$     \\ \hline
%    $2$ & $\left(0, \,  -\frac{t_1}{x} \right)$  &  $\left(\frac{t_1}{x}, \,  -1 \right)$  &  $\left\{ \left(t_2, \, -\frac{t_1}{x} \right)\right\}$        \\ \hline
%\end{tabular}
%\end{center}
%in which the third column will be updated during computation.
Let us compute respective $R$-pairs of
$$f= \frac { \left(  \left( x-1 \right)^{2}t_1+x \right) t_2^3 +x \left( x-1 \right)t_1}{x^2 \left( x-1 \right) 
t_2^2} \quad \text{and} \quad \tilde{f}=t_2^2.$$

First,   {\sc HermiteReduce}$(f)$ finds 
$(g, p, s)  \in K_1(t_2) \times K_1[t_2] \times S_2$
such that \eqref{EQ:hr} holds, where 
$g= \frac{1}{t_2},$ $p = \frac{(x-1)^2 t_1+x}{x^2(x-1)} t_2$ and $s = 0.$

Second, {\sc AuxiliaryReduction}$(p)$ yields 
$(q, r) \in K_1[t_2] \times A_2$
such that \eqref{EQ:ar} holds, where %$\phi_1: K_1 \rightarrow K_1$ is the complete reduction for $K_1^\prime$,
$q= \frac{t_1}{x} t_2 + \frac{x-1}{x} t_1^2$ and $r = \frac{t_1}{x} t_2 - \frac{2t_1}{x}.$

%To reduce $r$ further, we need $v_0$ and $v_1$ in $B_2$, because $r$ has
%degree 1 in $t_2$.
%Using $\phi_1$ and a variant of Algorithm \ref{ALG:basis}, we update $B_2$ to be
%\[ \left\{ \left(t_2, \, -\frac{t_1}{x} \right), \, \left( \frac{t_2^2}{2}, \, - \frac{t_1t_2}{x} \right) \right\}. \]
%Meanwhile, $B_1$ is also updated to be
%\[ \left\{ \left(t_1, \, \frac{1}{x-1} \right), \, \left( \frac{t_1^2}{2}, \, - \frac{t_1}{x-1} \right) \right\}. \]

At last, we project $r$ to $K_1[t_2]^\prime$ and the $\theta_2$-complement by {\sc Projection}. The respective projections are $u^\prime$ and
$0,$ where $u=-\frac{t_2^2}{2} +2 t_2$. So $f$ has an R-pair $(g+q+u, 0)$. Consequently,  $ \int f= g+q+u$.

In the same vein, an R-pair of $\tilde{f}$ is $(\tilde g, \tilde r)$, where
\[ \tilde g = {x t_2^{2} + \left( 2t_1x-2t_1-2x \right)t_2}+
2t_1^{2}x-2t_1^{2}-6 t_1 x+6 t_1+6x\]
and $\tilde r = -\frac{2 t_1^2}{x}.$
So $\tilde{f}$ does not have any integral in $K_2$. The remainder $ \tilde r$ is \lq\lq simpler\rq\rq\
than $\tilde{f}$ in the sense that $\tilde r$ is of degree $0$ in $t_2$. 
\end{example}
\begin{example}\label{EX:alg}
Let $K_0 = C(x,y)$ with $y^3 - xy + 1=0$. Set $t_1 = \log(y)$. %to be the logarithm of $y$. % = \log(y)$.
Then $K_1=K_0(t_1)$ is a primitive tower.
Two  associated  pairs of $K_1$ are  $(\lambda_1, \phi_{0}(t_1^\prime)) = \left(\frac{2xy}{3}, \,-y \right)$
and $(\theta_1, c_1) = \left(y, \, -1 \right)$, respectively. 
We compute an R-pair of $f=y(2-3t_1).$

{\sc HermiteReduce}$(f)$ finds a triplet
$(g, p, s)$ in $K_0(t_1) \times K_0[t_1] \times S_1$
such that \eqref{EQ:hr} holds, where
$g = 0,$ $p = -3yt_1+ 2y$ and $s = 0.$

Since $\phi_{0}(t_1^\prime) = -y$, we see that $y \in \im(\phi_0)$. Then $p \in A_1$.
So \eqref{EQ:ar} holds by setting $q=0$ and $r=p$. 

{\sc Projection}$(r, \lambda_1, \phi_{0}(t_1^\prime), \theta_1, c_1)$   yields
$u=  \frac{3}{2} t_1^2- (2xy)t_1+ 2xy$ and $v=0$
such that \eqref{EQ:proj} holds.
Thus, an R-pair of $f$ is $(u, 0)$.  Consequently, $u$ is an integral of $f$.
\end{example}

%\begin{example}
%We illustrate how remainders help
%us to solve the limited integration problem  (see \cite[\S 7.2]{BronsteinBook}):
%Given $f, w_1, \ldots, w_m \in K_n$, find $z_1, \ldots z_m \in C$ and $y \in K_n$ such that 
%\begin{equation} \label{EQ:limited}
%f = y^\prime + \sum_{i \in [m]} z_i w_i.
%\end{equation}
%The above equation corresponds to a parametric Risch equation for computing
%elementary integrals over primitive towers. 
%
%%It is a special form of parametric Risch equations  (see \cite[\S 7.2]{BronsteinBook}). 
%
%Applying $\phi_n$ to \eqref{EQ:limited} yields $\phi_n(f) = \sum_{i \in [m]} z_i \phi_n(w_i)$. This leads to
%a linear algebraic system $L$ in $z_1, \ldots, z_m$. If $L$ has no solution, neither does \eqref{EQ:limited}.
%Otherwise, let $\tilde{c}_1, \ldots, \tilde{c}_m \in C$ be a solution of $L$. A  byproduct of setting up $L$ consists of
%$g, u_1, \ldots, u_m \in K_n$ such that $f = g^\prime+\phi_n(f)$ and $w_i = u_i^\prime + \phi_n(w_i)$ with $i \in [m]$.
%So \eqref{EQ:limited} has a solution $z_1=\tilde{c}_1,$ \ldots, $z_m=\tilde{c}_m $ and $y= g - \sum_{i \in [m]} \tilde{c}_i u_i$.
%\end{example}

We have compared our preliminary implementation of the complete reduction given in Theorem \ref{TH:tower}
with the {\sc Maple} function {\tt int} and Algorithm {\sc AddDecompInField} in~\cite[page 150]{DGLW2020} for in-field integration. Empirical results are given in the appendix.

%Our empirical results are preliminary. The scripts for both {\tt CR} and {\tt AD} need to be optimized.

\section{Applications of remainders} \label{SECT:app}

This section contains
two applications: computing  elementary integrals over~$K_n$ with $K_0=C(x)$,
and constructing telescopers for some non-D-finite functions. Convention \ref{CON:tower} is
kept in the  sequel. 
\subsection{Elementary integrals} \label{SUBSECT:elem}

Let $f \in K_n$. Then $f$ has an elementary integral over $K_n$ if and only if its remainder 
$\phi_n(f)$ has one. Two properties of remainders allow us to apply Algorithm \ref{ALG:cs} directly
to compute elementary integrals. To describe the properties, we need three $C$-subspaces of $K_n$.
 Let
$$ P = \sum_{i \in [n]} t_i K_{i-1}[t_i], \quad S= \sum_{i \in [n]} S_i,$$
and $T$ be the  $C$-subspace spanned by 
$\phi_0(t_1^\prime),$ $\phi_0(t_2^\prime),$ \ldots, $\phi_{n-1}(t_n^\prime).$ 
Note that 
$\sum_{i \in [n]} t_i K_{i-1}[t_i]$,  $\sum_{i \in [n]} S_i$ 
and $K_0+P+S$ are all  direct. 
\begin{prop} \label{PROP:rem1}
$ \im(\phi_n) \subset K_0 \oplus P \oplus S.$
\end{prop}
\begin{proof} The conclusion holds for $n=0$ because $\im(\phi_0) \subset K_0$.
Assume that $n>0$ and that the conclusion holds for $n-1$. By
Theorem \ref{TH:decomp}, $\im(\phi_n) \subset A_n + S_n$.
Since $A_n \subset \im(\phi_{n-1}) + t_n K_{n-1}[t_n]$,
we see that  
$\im(\phi_n) \subset \im(\phi_{n-1}) + t_n K_{n-1}[t_n] + S_n.$ 
The proposition then follows from the induction hypothesis. 
\end{proof}
\begin{prop}\label{PROP:rem2}
If $h \in K_0 \oplus S$, then
$h - \phi_n(h) \in K_0^\prime  + T.$
\end{prop}
\begin{proof} Assume $h = h_0 + \sum_{i \in [n]} s_i$, where $h_0 \in K_0$ and $s_i \in  S_i$.
Then $s_i = \phi_i(s_i)$ by Theorem \ref{TH:decomp}, and $\phi_i(s_i) \equiv \phi_n(s_i)$ $\text{mod}~T$ 
by Corollary \ref{COR:tower}. Hence, 
$ s_i \equiv \phi_n(s_i)$ $\text{mod}~T$, which, together with the application of $\phi_n$ to $h$, implies   
$h - \phi_n(h) \equiv h_0 - \phi_n(h_0)$ $\text{mod}~T.$
By Corollary \ref{COR:tower} again, 
$h - \phi_n(h) \equiv h_0 - \phi_0(h_0)$ $\text{mod}~T.$
The proposition is proved by noting that $h_0 - \phi_0(h_0) \in K_0^\prime$. 
\end{proof}

An element $s$ of $S$ can be uniquely written as $\sum_{i \in [n]} s_i$,
where $s_i \in S_i$. We say that all residues of $s$ are constants if all residues of $s_i$ as an element in $K_{i-1}(t_i)$ belong to $\overline{C}$ 
for every $i \in [n]$.

\begin{thm}\label{TH:elem}
Let $K_n$ be  a primitive tower  as in \eqref{EQ:tower} with $K_0=C(x)$. 
Assume that $C$ is algebraically closed.
%Let $\phi_i$ be the same as in Convention \ref{CON:tower}.  
Then $f \in K_n$ has an elementary integral over $K_n$ if and only if 
\begin{itemize}
  \item [(i)] there exists $s\in S$ such that
  $\phi_n(f)\equiv s \mod K_0 + T,$ and
  \item [(ii)] all residues of $s$ belong to $C$.
\end{itemize}
\end{thm}
\begin{proof} Assume that both (i) and (ii) hold.  By (ii) and \cite[Proposition 3.3]{DGGL2023},
$s$ has an elementary integral over $K_n$.
Every element of $K_0$ has an elementary integral over $K_0$ because $K_0=C(x)$.
By Remark \ref{RE:lc}, $T \subset K_n^\prime$. It follows from (i) that  
$\phi_n(f)$ has an elementary integral over $K_n$, and so does $f$.

Conversely, assume that $f$ has an elementary integral over $K_n$.
Then there exists a $C$-linear combination $h$ of logarithmic derivatives in $K_n$ such that 
  $f \equiv h \mod K_n^\prime$  by \cite[Theorem 5.5.2]{BronsteinBook}.
  Since $\phi_n(f)=\phi_n(h)$, it suffices to show that $\phi_n(h)$ satisfies both (i) and (ii).
  By the logarithmic derivative identity, $h \equiv s$ $\text{mod}$ $K_0$
  for some $s \in S$, which has merely constant residues. 
  Then $h \equiv \phi_n(h)$ $\text{mod}$ $K_0 + T$ by Proposition \ref{PROP:rem2}.
  Hence, $\phi_n(h) \equiv s$ $\text{mod}~K_0 + T$ by the above two congruences.
  Both (i) and (ii) hold. 
  \end{proof}

Next, we outline an algorithm for computing elementary integrals over $K_n$.
Let $f \in K_n$.
\begin{enumerate}
\item[1.] Compute an R-pair $(g, \phi_n(f))$. 
If $\phi_n(f)=0$, then $\int f = g$ and we are done.
\item[2.] Assume that $\phi_n(f) \neq 0$. By Proposition \ref{PROP:rem1}, we can write
$\phi_n(f) = r + p + s$ and $\phi_{i-1}(t_i^\prime) = r_i + p_i + s_i$,
where $i \in [n]$, $r, r_i \in K_0,$  $p, p_i  \in P$ and $s, s_i \in S$.
\item[3.] Let $z_1, \ldots, z_n$ be constant indeterminates.
\begin{itemize}
\item[-] Use {\sc ConstantMatrix} (Algorithm \ref{ALG:cs})
to compute a matrix $M \in C^{k \times n}$ and $\vv \in C^k$
such that $s - \sum_{i \in [n]} z_i s_i$ has merely constant residues if and only if the linear system given by the augmented matrix $(M,\vv)$ is consistent.
\item[-] Compute $N \in C^{l \times n}$ and $\vw \in C^l$ such that $p = \sum_{i \in [n]} z_i p_i$
if and only if the linear system given by the augmented matrix $(N,\vw)$ is consistent.
\item[-] Solve the linear system
$\begin{pmatrix}
      M \\
      N
      \end{pmatrix}\begin{pmatrix}
        z_1, \ldots, z_n
        \end{pmatrix}^\tau
       = \begin{pmatrix}
        \vv \\
        \vw
       \end{pmatrix}.$
\end{itemize}
\item[4.] If the above system has no solution, then $f$ has no elementary integral over $K_n$ by Theorem \ref{TH:elem}.
Otherwise, let $\tilde{c}_1, \ldots, \tilde{c}_n$ be such a solution. Set
$\tilde{r} = r - \sum_{i \in [n]} \tilde{c}_i r_i$ and $
   \tilde{s} = s - \sum_{i \in [n]} \tilde{c}_i s_i.$ Then
$\int f = g + \int \tilde{r} + \int \tilde{s} + \sum_{i \in [n]} \tilde{c}_i (t_i -\lambda_i).$
Note that $\int \tilde{r}$ is elementary because $\tilde{r} \in  C(x)$,
and that $\int \tilde{s}$ is elementary over $K_n$ by Theorem \ref{TH:elem}. 
\end{enumerate}
\begin{example} \label{EX:elem}
We follow the above outline to integrate 
$$f= \frac{x+(x-1)t_2}{(x-1)t_1}+\frac{t_2+t_3(1-t_1)}{x},$$
\begin{enumerate}
\item[1.] By $\phi_3$, we find an R-pair $(t_2t_3, \phi_3(f))$, 
where 
$\phi_3(f)=\frac{x}{(x-1)t_1}.$
\item[2.] Compute $\phi_{i-1}(t_i^\prime) = r_i + p_i + s_i,$ %and $\tilde{\phi}_{i-1}(\tilde{t}_i^\prime)
%{=} \tilde{r}_i + \tilde{p}_i + \tilde{s}_i$, 
where 
\begin{center}
\begin{tabular}{|c|c|c|c|} \hline 
$i$ & $1$ & $2$ & $3$ \\ \hline
$(r_i, p_i, s_i)$ & $\left(\frac{1}{x-1}, 0, 0 \right)$ & $\left(\frac{1}{x}, -\frac{t_1}{x}, 0 \right)$ & $\left(\frac{1}{x}, 0 , \frac{1}{t_1} \right)$ 
\\ \hline
\end{tabular}
\end{center}
\item[3.] By step 3 in the above outline,  we have 
 $$
\begin{pmatrix}
 0 &1&  0 \\
     0 & 0 & -1  
      \end{pmatrix}
      \begin{pmatrix}
        z_1 \\ z_2 \\ z_3 
        \end{pmatrix}
       = \begin{pmatrix}  0 \\ -1  \end{pmatrix}.         
$$     
It has a solution $z_1=z_2=0$ and $z_3=1$.
\item[4.] Computing the residues yields  $\int f =t_2t_3+t_3+\log\left(\frac{t_1}{x}\right)$. 
 \end{enumerate}
\end{example}
%The {\tt int()} command in {\sc Maple 2021} returned unevaluated integrals for both $f$ and $\tilde{f}$ in this example.
Neither {\tt int()} command in {\sc Maple 2021} nor {\tt Integrate[]}  command in {\sc Mathematica 14.1} found an elementary integral for $f$. The {\sc Axiom}-based computer algebra system {\sc FriCAS 1.3.10} (see \cite{FriCAS})
returned a correct integral.
 Comprehensive tests are given in \cite{Abbasi2024}  for elementary integration in  current computer algebra systems.

\subsection{Telescopers}

General connections between symbolic integration and creative-telescoping
are described in \cite[Chapter 1]{RaabThesis}. Examples in \cite[\S 7]{CDL2018} illustrate that additive decompositions help us detect the existence  of telescopers for elements in some primitive towers.  We present two propositions for the same purpose by remainders and residues.

Let $K=C(x,y)$ be the field of rational functions in $x$ and $y$
equipped with the usual partial derivatives $D_x$ and $D_y$.
Differential fields related to integration for several derivations can be found in \cite{CavinessRothstein, Rosenlicht1976}. 
Let $t$ be an element in some partial differential field extension of $K$ such that $t$ is transcendental 
over $K$, $D_yD_x(t)=D_xD_y(t),$  $D_x(t) \in K[t]$ with degree less than two, and $D_y(t) \in K \setminus D_y(K)$.
Then $t$ is a primitive monomial over $K$ with respect to $D_y$.
The extended derivatives are still denoted by $D_x$ and $D_y$, respectively.

Every element of $K[t]$ is D-finite over $K$.
But $K(t)$ contains non-D-finite elements. For instance, 
$t^{-1}$ is not D-finite over $K$, because $t^{i+1}$ is the monic denominator of $D_y^i(t^{-1})$ for all $i \in \bN$. 

For $f \in K(t)$, 
a differential operator $L \in C(x)[D_x]^\times$ is called a {\em telescoper} for $f$
if $L(f) \in  D_y(K(t))$.

\begin{prop}\label{PROP:telescoper1}
Let $\phi: K(t) \rightarrow K(t)$ be the complete reduction for $D_y(K(t))$ given in Convention \ref{CON:tower} with $K=K_0$ and $\phi=\phi_1$.
For $f \in K(t)$ and $m \in \bN$, $f$ has a telescoper of order no more than $m$ if and only if
there exist $l_0$ \ldots, $l_m \in C(x)$, not all zero,
such that 
\begin{equation} \label{EQ:linear}
\sum_{i \in [m]_0} l_i \phi(D_x^i(f))  = 0.
\end{equation} 
\end{prop}
\begin{proof} Let $L = \sum_{i \in [m]_0} l_i D_x^i$
with $l_0, \ldots, l_m \in C(x)$, not all zero.  Then 
$
\phi(L(f)) = \sum_{i \in [m]_0} l_i \phi(D_x^i(f)), 
$
because $\phi$ is $C(x)$-linear.
Assume that \eqref{EQ:linear} holds. Then $L$ is a telescoper for $f$ with order no more than $m$.
Conversely, assume that $L$ is a telescoper for $f$ with order no more than $m$. Then $\phi(L(f))=0$
because $\phi$ is a complete reduction.
Hence, \eqref{EQ:linear} holds. 
\end{proof}
Below is a sufficient condition on the existence of telescopers.   
\begin{prop} \label{PROP:telescoper2}
Let $f \in K(t)$. Then there exists a unique element $s \in S_t$ such that $\phi(f) \equiv s$ $\text{mod}$ $K[t]$.
If all residues of $s$ with respect to $D_y$ are in $\overline{C(x)}$,
then $f$ has a telescoper. 
\end{prop}
\begin{proof}  There exists a unique pair $(q, s)$ in $K[t] \times S_t$
such that $\phi(f)=q+s$ by Proposition \ref{PROP:rem1}. 
Since $q$ is D-finite over $K$, it has a telescoper by \cite[Lemma 4.1]{Zeil1990} or \cite[Lemma 3]{Lips1988}. 

It remains to prove that $s$ has a telescoper by \cite[Remark 2.3]{CHLW2016}. 
Let $s= \frac{a}{b}$, where $a, b \in K[t]$, $b$ is monic with respect to $t$ and $\gcd(a,b)=1$. 
Assume that $\alpha_1, \ldots, \alpha_k$ are the distinct roots of $b$. 
By \cite[Lemma 3.1 (i)]{DGGL2023}, we have that 
\begin{equation} \label{EQ:residue}
s = \sum_{j \in [k]}  \beta_j \frac{D_y(t-\alpha_j)}{t-\alpha_j},
\end{equation}
where $\beta_j \in \overline{K}$ is the residue of $f$ at $\alpha_j$ with respect to $D_y$. 
Since each $\beta_j$ is assumed to be in $\overline{C(x)}$,
there exists $L \in C(x)[D_x]$ annihilating all of them by \cite[Theorem 3.29 (3)]{KauersBook}. 
By the commutativity of applying derivations and taking logarithmic derivatives, we have
$$
D_x \left(\gamma \frac{D_y(u)}{u} \right) =  D_y\left(\gamma \frac{ D_x(u)}{u} \right) + D_x(\gamma) \frac{D_y(u)}{u}.
$$
for all $\gamma \in \overline{C(x)}$ and $u \in K(t)$. 
A repeated application of the above equality to \eqref{EQ:residue}  yields $g \in \overline{C(x)}(y, t)$ such that
$$ L(s) = D_y(g) +  \sum_{j \in [k]}  L\left( \beta_j \right)   \frac{D_y(t-\alpha_j)}{t-\alpha_j} = D_y(g) $$
Moreover, $g$ is symmetric in $\alpha_1$, \ldots $\alpha_k$ over $K(t)$ so that $g$ actually belongs to $K(t)$. 
\end{proof}

\begin{example} \label{EX:telesoper}
Let $K= \bC(x,y)$ and $t=\log(x+y)$. We try to construct respective telescopers for 
\[ f=\frac{2x}{(x+y)(t^2-x)} \quad \text{and} \quad \tilde{f} = 
y \frac{D_y(t-y)}{t-y}. \]

Note that $f$ is simple. So $\phi(f)=f$.  Its nonzero residues are $\pm \sqrt{x}$ by \cite[Theorem 4.4.3]{BronsteinBook}. 
By Proposition \ref{PROP:telescoper2},  $f$ has a telescoper.
Using the notation in Proposition \ref{PROP:telescoper1}, we have 
$2 x \phi(D_x(f)) = f.$
Thus, the minimal telescoper for $f$ 
is $2xD_x - 1$. 

Again, $\tilde f$ is simple. So $\phi(\tilde{f})= \tilde{f}$. 
Since $\tilde{f}$ has a nonzero residue  $y$, Proposition \ref{PROP:telescoper2} is not applicable.
Let 
$g = \frac{D_y(t-y)}{t-y}$ and $\gamma = \frac{D_x(t-y)}{D_y(t-y)}.$
Then $\tilde{f} = y g$ and $\gamma = (1-x-y)^{-1}$. For $\omega \in C(x,y)$, we calculate 
\begin{align*}
D_x \left( \omega g \right) & = D_x(\omega) g + \omega D_x \left( g \right) 
 = D_x(\omega) g + \omega D_y \left(\frac{D_x(t-y)}{t-y}\right) \\
& \equiv D_x(\omega) g  - D_y(\omega) \frac{D_x(t-y)}{t-y}  \mod D_y(K(t)) \\
&  \equiv \left( D_x(\omega) - \gamma D_y(\omega) \right) g \mod D_y(K(t)).
\end{align*}
Then $\phi(D_x(\omega g)) = \left( D_x(\omega) - \gamma D_y(\omega) \right) g$ because $g$ is simple.
Set $\gamma_0=y$ and $\gamma_i = D_x(\gamma_{i-1}) - \gamma D_y(\gamma_{i-1})$ for $i \ge 1$. 
It follows from  the above calculation that 
$\phi( D_x^i (\tilde{f}) )= \gamma_i g.$ 
Moreover, the denominator of $\gamma_i$ has degree $2i-1$ in $y$ for $i \ge 1$ by a straightforward induction. Therefore, 
$\phi(\tilde{f})$, $\phi(D_x(\tilde{f}))$, $\phi(D_x^2(\tilde{f}))$, \ldots
are linearly independent over $C(x)$. Consequently,  $\tilde{f}$ has no telescoper by Proposition \ref{PROP:telescoper1}.  
\end{example}
\section{Conclusions} \label{SECT:conc}
In this article, we have developed a complete reduction for derivatives in a primitive tower.
The reduction algorithm decomposes an element of such a tower as the sum of a derivative and a remainder,
where the derivative is unique up to an additive constant and the remainder is unique.
The algorithm can be applied to compute elementary integrals over primitive towers and to construct telescopers for some non-D-finite functions.
The work is a step forward in  the development of complete reductions for derivatives in transcendental Liouvillian extensions.

\begin{acks}
We thank Shaoshi Chen, Manuel Kauers, Peter Paule, Clemens Raab and Carsten Schneider for valuable discussions and suggestions.

 Special thanks go to Nasser M.\ Abbasi and Ralf Hemmecke for helping us experiment with {\sc FriCAS}
and sharing their experience in symbolic integration.
Part of the calculation in Example \ref{EX:alg} was carried out by 
the prototype implementation of lazy Hermite reduction and creative telescoping for algebraic functions
by Shaoshi Chen, Lixin Du and Manuel Kauers.

We are grateful to the anonymous referees for their friendly and careful reviews. Their comments 
guided us to revise the abstract and introduction substantially, and to decrease the complexity for constructing the $C$-basis described in Lemma \ref{LM:basis}.  
\end{acks}

%\newpage 

\bibliographystyle{ACM-Reference-Format}
%\bibliographystyle{plain}
%\balance
\bibliography{primitive}


\begin{thebibliography}{36}

%%% the Web address.

\ifx \showCODEN    \undefined \def \showCODEN     #1{\unskip}     \fi
\ifx \showDOI      \undefined \def \showDOI       #1{#1}\fi
\ifx \showISBNx    \undefined \def \showISBNx     #1{\unskip}     \fi
\ifx \showISBNxiii \undefined \def \showISBNxiii  #1{\unskip}     \fi
\ifx \showISSN     \undefined \def \showISSN      #1{\unskip}     \fi
\ifx \showLCCN     \undefined \def \showLCCN      #1{\unskip}     \fi
\ifx \shownote     \undefined \def \shownote      #1{#1}          \fi
\ifx \showarticletitle \undefined \def \showarticletitle #1{#1}   \fi
\ifx \showURL      \undefined \def \showURL       {\relax}        \fi
% invisible to TeX
\providecommand\bibfield[2]{#2}
\providecommand\bibinfo[2]{#2}
\providecommand\natexlab[1]{#1}
\providecommand\showeprint[2][]{arXiv:#2}

\bibitem[\protect\citeauthoryear{Abbasi}{Abbasi}{2024}]%
        {Abbasi2024}
\bibfield{author}{\bibinfo{person}{Nasser~M. Abbasi}.}
  \bibinfo{year}{2024}\natexlab{}.
\newblock \bibinfo{title}{Computer Algebra Independent Integration Tests}.
\newblock
\newblock
\newblock
\shownote{Available at
  \url{https://12000.org/my\_notes/CAS\_integration\_tests/index.htm}.}


\bibitem[\protect\citeauthoryear{Abramov, Geddes, and Le}{Abramov
  et~al\mbox{.}}{2002}]%
        {AGL2002}
\bibfield{author}{\bibinfo{person}{Sergei~A. Abramov},
  \bibinfo{person}{Keith~O. Geddes}, {and} \bibinfo{person}{Ha~Q. Le}.}
  \bibinfo{year}{2002}\natexlab{}.
\newblock \showarticletitle{Computer algebra library for the construction of
  the minimal telescopers}.
\newblock In \bibinfo{booktitle}{\emph{Mathematical Software ({B}eijing,
  2002)}}. \bibinfo{publisher}{World Sci. Publ.}, \bibinfo{address}{River Edge,
  NJ}, \bibinfo{pages}{319--329}.
\newblock
\showISBNx{981-238-048-5}


\bibitem[\protect\citeauthoryear{Abramov and Petkov{\v s}ek}{Abramov and
  Petkov{\v s}ek}{2001}]%
        {AP2001}
\bibfield{author}{\bibinfo{person}{Sergei~A. Abramov} {and}
  \bibinfo{person}{Marko Petkov{\v s}ek}.} \bibinfo{year}{2001}\natexlab{}.
\newblock \showarticletitle{Minimal decomposition of indefinite hypergeometric
  sums}. In \bibinfo{booktitle}{\emph{Proceedings of the 26th International
  Symposium on Symbolic and Algebraic Computation ({ISSAC}'01)}}.
  \bibinfo{publisher}{ACM}, \bibinfo{address}{New York, NY, USA},
  \bibinfo{pages}{7--14}.
\newblock


\bibitem[\protect\citeauthoryear{Abramov and Petkov{\v s}ek}{Abramov and
  Petkov{\v s}ek}{2002}]%
        {AP2002}
\bibfield{author}{\bibinfo{person}{Sergei~A. Abramov} {and}
  \bibinfo{person}{Marko Petkov{\v s}ek}.} \bibinfo{year}{2002}\natexlab{}.
\newblock \showarticletitle{Rational normal forms and minimal decompositions of
  hypergeometric terms}.
\newblock \bibinfo{journal}{\emph{Journal of Symbolic Computation}}
  \bibinfo{volume}{33}, \bibinfo{number}{5} (\bibinfo{year}{2002}),
  \bibinfo{pages}{521--543}.
\newblock


\bibitem[\protect\citeauthoryear{Bostan, Chen, Chyzak, Li, and Xin}{Bostan
  et~al\mbox{.}}{2013}]%
        {BCCLX2013}
\bibfield{author}{\bibinfo{person}{Alin Bostan}, \bibinfo{person}{Shaoshi
  Chen}, \bibinfo{person}{Fr{\'e}d{\'e}ric Chyzak}, \bibinfo{person}{Ziming
  Li}, {and} \bibinfo{person}{Guoce Xin}.} \bibinfo{year}{2013}\natexlab{}.
\newblock \showarticletitle{Hermite reduction and creative telescoping for
  hyperexponential functions}. In \bibinfo{booktitle}{\emph{Proceedings of the
  38th International Symposium on Symbolic and Algebraic Computation
  ({ISSAC}'13)}}. \bibinfo{publisher}{ACM}, \bibinfo{address}{New York, NY,
  USA}, \bibinfo{pages}{77--84}.
\newblock
\showISBNx{978-1-4503-2059-7}


\bibitem[\protect\citeauthoryear{Bostan, Chyzak, Lairez, and Salvy}{Bostan
  et~al\mbox{.}}{2018}]%
        {BCPS2018}
\bibfield{author}{\bibinfo{person}{Alin Bostan},
  \bibinfo{person}{Fr{\'e}d{\'e}ric Chyzak}, \bibinfo{person}{Pierre Lairez},
  {and} \bibinfo{person}{Bruno Salvy}.} \bibinfo{year}{2018}\natexlab{}.
\newblock \showarticletitle{Generalized Hermite reduction, creative telescoping
  and definite integration of {D}-finite functions}. In
  \bibinfo{booktitle}{\emph{Proceedings of the 43th International Symposium on
  Symbolic and Algebraic Computation ({ISSAC}'18)}}. \bibinfo{publisher}{ACM},
  \bibinfo{address}{New York, NY, USA}, \bibinfo{pages}{95--102}.
\newblock


\bibitem[\protect\citeauthoryear{Boulier, Lemaire, Lallemand, Regensburger, and
  Rosenkranz}{Boulier et~al\mbox{.}}{2016}]%
        {BLLRR2016}
\bibfield{author}{\bibinfo{person}{Fran\c{c}ois Boulier},
  \bibinfo{person}{Fran\c{c}ois Lemaire}, \bibinfo{person}{Joseph Lallemand},
  \bibinfo{person}{Georg Regensburger}, {and} \bibinfo{person}{Markus
  Rosenkranz}.} \bibinfo{year}{2016}\natexlab{}.
\newblock \showarticletitle{Additive normal forms and integration of
  differential fractions}.
\newblock \bibinfo{journal}{\emph{Journal of Symbolic Computation}}
  \bibinfo{volume}{77} (\bibinfo{year}{2016}), \bibinfo{pages}{16--38}.
\newblock


\bibitem[\protect\citeauthoryear{Bronstein}{Bronstein}{2005}]%
        {BronsteinBook}
\bibfield{author}{\bibinfo{person}{Manuel Bronstein}.}
  \bibinfo{year}{2005}\natexlab{}.
\newblock \bibinfo{booktitle}{\emph{Symbolic {I}ntegration {I}: Transcendental
  Functions} (\bibinfo{edition}{2nd} ed.)}.
\newblock \bibinfo{publisher}{Springer-Verlag}, \bibinfo{address}{Berlin}.
\newblock
\showISBNx{3-540-21493-3}


\bibitem[\protect\citeauthoryear{Caviness and Rothstein}{Caviness and
  Rothstein}{1975}]%
        {CavinessRothstein}
\bibfield{author}{\bibinfo{person}{Bob~F. Caviness} {and}
  \bibinfo{person}{Michael Rothstein}.} \bibinfo{year}{1975}\natexlab{}.
\newblock \showarticletitle{A {L}iouville theorem on integration in finite
  terms for line integrals}.
\newblock \bibinfo{journal}{\emph{Communications in Algebra}}
  \bibinfo{volume}{3} (\bibinfo{year}{1975}), \bibinfo{pages}{781--795}.
\newblock


\bibitem[\protect\citeauthoryear{Chen, Du, Gao, Huang, Li, and Li}{Chen
  et~al\mbox{.}}{2025}]%
        {CDGHLL2025}
\bibfield{author}{\bibinfo{person}{Shaoshi Chen}, \bibinfo{person}{Hao Du},
  \bibinfo{person}{Yiman Gao}, \bibinfo{person}{Hui Huang},
  \bibinfo{person}{Wenqiao Li}, {and} \bibinfo{person}{Ziming Li}.}
  \bibinfo{year}{2025}\natexlab{}.
\newblock \bibinfo{title}{A complete reduction in transcendental Liouvillian
  extensions}.
\newblock
\newblock
\newblock
\shownote{In preparation.}


\bibitem[\protect\citeauthoryear{Chen, Du, and Li}{Chen et~al\mbox{.}}{2018a}]%
        {CDL2018}
\bibfield{author}{\bibinfo{person}{Shaoshi Chen}, \bibinfo{person}{Hao Du},
  {and} \bibinfo{person}{Ziming Li}.} \bibinfo{year}{2018}\natexlab{a}.
\newblock \showarticletitle{Additive decompositions in primitive extensions}.
  In \bibinfo{booktitle}{\emph{Proceedings of the 43th {I}nternational
  {S}ymposium on {S}ymbolic and {A}lgebraic {C}omputation ({ISSAC}'18)}}.
  \bibinfo{publisher}{ACM}, \bibinfo{address}{New York, NY, USA},
  \bibinfo{pages}{135--142}.
\newblock


\bibitem[\protect\citeauthoryear{Chen, Du, and Kauers}{Chen
  et~al\mbox{.}}{2021}]%
        {CDK2021}
\bibfield{author}{\bibinfo{person}{Shaoshi Chen}, \bibinfo{person}{Lixin Du},
  {and} \bibinfo{person}{Manuel Kauers}.} \bibinfo{year}{2021}\natexlab{}.
\newblock \showarticletitle{Lazy {H}ermite reduction and creative telescoping
  for algebraic functions}. In \bibinfo{booktitle}{\emph{Proceedings of the
  46th {I}nternational {S}ymposium on {S}ymbolic and {A}lgebraic {C}omputation
  ({ISSAC}'21)}}. \bibinfo{publisher}{ACM}, \bibinfo{address}{New York, NY,
  USA}, \bibinfo{pages}{75--82}.
\newblock


\bibitem[\protect\citeauthoryear{Chen, Du, and Kauers}{Chen
  et~al\mbox{.}}{2023}]%
        {CDK2023}
\bibfield{author}{\bibinfo{person}{Shaoshi Chen}, \bibinfo{person}{Lixin Du},
  {and} \bibinfo{person}{Manuel Kauers}.} \bibinfo{year}{2023}\natexlab{}.
\newblock \showarticletitle{Hermite reduction for {D}-finite functions via
  integral bases}.
\newblock In \bibinfo{booktitle}{\emph{Proceedings of the 48th {I}nternational
  {S}ymposium on {S}ymbolic and {A}lgebraic {C}omputation ({ISSAC}'23)}}.
  \bibinfo{publisher}{ACM}, \bibinfo{address}{New York, NY, USA},
  \bibinfo{pages}{155--163}.
\newblock


\bibitem[\protect\citeauthoryear{Chen, Hou, Labahn, and Wang}{Chen
  et~al\mbox{.}}{2016a}]%
        {CHLW2016}
\bibfield{author}{\bibinfo{person}{Shaoshi Chen}, \bibinfo{person}{Qinghu Hou},
  \bibinfo{person}{George Labahn}, {and} \bibinfo{person}{Ronghua Wang}.}
  \bibinfo{year}{2016}\natexlab{a}.
\newblock \showarticletitle{Existence problem of telescopers: beyond the
  bivarate case}. In \bibinfo{booktitle}{\emph{Proceedings of the 41th
  {I}nternational {S}ymposium on {S}ymbolic and {A}lgebraic {C}omputation
  ({ISSAC}'16)}}. \bibinfo{publisher}{ACM}, \bibinfo{address}{New York, NY,
  USA}, \bibinfo{pages}{167--174}.
\newblock


\bibitem[\protect\citeauthoryear{Chen, Kauers, and Koutschan}{Chen
  et~al\mbox{.}}{2016b}]%
        {CKK2016}
\bibfield{author}{\bibinfo{person}{Shaoshi Chen}, \bibinfo{person}{Manuel
  Kauers}, {and} \bibinfo{person}{Christoph Koutschan}.}
  \bibinfo{year}{2016}\natexlab{b}.
\newblock \showarticletitle{Reduction-based creative telescoping for algebraic
  functions}.
\newblock In \bibinfo{booktitle}{\emph{Proceedings of the 41th {I}nternational
  {S}ymposium on {S}ymbolic and {A}lgebraic {C}omputation ({ISSAC}'16)}}.
  \bibinfo{publisher}{ACM}, \bibinfo{address}{New York, NY, USA},
  \bibinfo{pages}{117--124}.
\newblock


\bibitem[\protect\citeauthoryear{Chen, van Hoeij, Kauers, and Koutschan}{Chen
  et~al\mbox{.}}{2018b}]%
        {CHKK2018}
\bibfield{author}{\bibinfo{person}{Shaoshi Chen}, \bibinfo{person}{Mark van
  Hoeij}, \bibinfo{person}{Manuel Kauers}, {and} \bibinfo{person}{Christoph
  Koutschan}.} \bibinfo{year}{2018}\natexlab{b}.
\newblock \showarticletitle{Reduction-based creative telescoping for {F}uchsian
  {D}-finite functions}.
\newblock \bibinfo{journal}{\emph{Journal of Symbolic Computation}}
  \bibinfo{volume}{85} (\bibinfo{year}{2018}), \bibinfo{pages}{108--127}.
\newblock


\bibitem[\protect\citeauthoryear{Davenport}{Davenport}{1986}]%
        {Davenport1986}
\bibfield{author}{\bibinfo{person}{James~H. Davenport}.}
  \bibinfo{year}{1986}\natexlab{}.
\newblock \showarticletitle{The {R}isch differential equation problem}.
\newblock \bibinfo{journal}{\emph{SIAM J. Comput.}}  \bibinfo{volume}{15}
  (\bibinfo{year}{1986}), \bibinfo{pages}{903--918}.
\newblock
\showISSN{0097-5397}


\bibitem[\protect\citeauthoryear{Du, Gao, Guo, and Li}{Du
  et~al\mbox{.}}{2023}]%
        {DGGL2023}
\bibfield{author}{\bibinfo{person}{Hao Du}, \bibinfo{person}{Yiman Gao},
  \bibinfo{person}{Jing Guo}, {and} \bibinfo{person}{Ziming Li}.}
  \bibinfo{year}{2023}\natexlab{}.
\newblock \showarticletitle{Computing logarithmic parts by evaluation
  homomorphisms}. In \bibinfo{booktitle}{\emph{Proceedings of the 48th
  {I}nternational {S}ymposium on {S}ymbolic and {A}lgebraic {C}omputation
  ({ISSAC}'23)}}. \bibinfo{publisher}{ACM}, \bibinfo{address}{New York, NY,
  USA}, \bibinfo{pages}{242--250}.
\newblock


\bibitem[\protect\citeauthoryear{Du, Guo, Li, and Wong}{Du
  et~al\mbox{.}}{2020}]%
        {DGLW2020}
\bibfield{author}{\bibinfo{person}{Hao Du}, \bibinfo{person}{Jing Guo},
  \bibinfo{person}{Ziming Li}, {and} \bibinfo{person}{Elaine Wong}.}
  \bibinfo{year}{2020}\natexlab{}.
\newblock \showarticletitle{An additive decomposition in logarithmic towers and
  beyond}. In \bibinfo{booktitle}{\emph{Proceedings of the 45th International
  Symposium on Symbolic and Algebraic Computation ({ISSAC}'20)}}.
  \bibinfo{publisher}{ACM}, \bibinfo{address}{New York, NY, USA},
  \bibinfo{pages}{146--153}.
\newblock


\bibitem[\protect\citeauthoryear{{FriCAS team}}{{FriCAS team}}{2024}]%
        {FriCAS}
\bibfield{author}{\bibinfo{person}{{FriCAS team}}.}
  \bibinfo{year}{2024}\natexlab{}.
\newblock \bibinfo{title}{{FriCAS} 1.3.10---an advanced computer algebra
  system}.
\newblock
\newblock
\newblock
\shownote{Available at \url{http://fricas.github.io}.}


\bibitem[\protect\citeauthoryear{Gao}{Gao}{2024}]%
        {GaoPhD2024}
\bibfield{author}{\bibinfo{person}{Yiman Gao}.}
  \bibinfo{year}{2024}\natexlab{}.
\newblock \emph{\bibinfo{title}{Additive {D}ecomposition and {E}lementary
  {I}ntegration over {E}xponential {E}xtensions}}.
\newblock \bibinfo{thesistype}{Ph.D. Dissertation}. \bibinfo{school}{Chinese
  Academy of Sciences}, \bibinfo{address}{Beijing, China}.
\newblock


\bibitem[\protect\citeauthoryear{Karr}{Karr}{1981}]%
        {Karr1981}
\bibfield{author}{\bibinfo{person}{Michael Karr}.}
  \bibinfo{year}{1981}\natexlab{}.
\newblock \showarticletitle{Summation in finite terms}.
\newblock \bibinfo{journal}{\emph{Journal of the Association for Computing
  Machinery}} \bibinfo{volume}{28}, \bibinfo{number}{2} (\bibinfo{year}{1981}),
  \bibinfo{pages}{305–350}.
\newblock


\bibitem[\protect\citeauthoryear{Kauers}{Kauers}{2023}]%
        {KauersBook}
\bibfield{author}{\bibinfo{person}{Manuel Kauers}.}
  \bibinfo{year}{2023}\natexlab{}.
\newblock \bibinfo{booktitle}{\emph{{D}-{F}inite {F}unctions}}.
\newblock \bibinfo{publisher}{Springer Cham}, \bibinfo{address}{Switzerland}.
\newblock


\bibitem[\protect\citeauthoryear{Le}{Le}{2003}]%
        {LeThesis}
\bibfield{author}{\bibinfo{person}{Ha~Q. Le}.} \bibinfo{year}{2003}\natexlab{}.
\newblock \emph{\bibinfo{title}{Algorithms for the Construction of the Minimal
  Telescopers}}.
\newblock \bibinfo{thesistype}{Ph.D. Dissertation}. \bibinfo{school}{University
  of Waterloo}, \bibinfo{address}{Waterloo, Canada, Ontariao}.
\newblock


\bibitem[\protect\citeauthoryear{Lipshitz}{Lipshitz}{1988}]%
        {Lips1988}
\bibfield{author}{\bibinfo{person}{Leonard Lipshitz}.}
  \bibinfo{year}{1988}\natexlab{}.
\newblock \showarticletitle{The diagonal of a {$D$}-finite power series is
  {$D$}-finite}.
\newblock \bibinfo{journal}{\emph{Journal of Algebra}} \bibinfo{volume}{113},
  \bibinfo{number}{2} (\bibinfo{year}{1988}), \bibinfo{pages}{373--378}.
\newblock
\showCODEN{JALGA4}
\showISSN{0021-8693}


\bibitem[\protect\citeauthoryear{Raab}{Raab}{2012a}]%
        {RaabThesis}
\bibfield{author}{\bibinfo{person}{Clemens~G. Raab}.}
  \bibinfo{year}{2012}\natexlab{a}.
\newblock \emph{\bibinfo{title}{Definite Integration in Differential Fields}}.
\newblock \bibinfo{thesistype}{Ph.D. Dissertation}. \bibinfo{school}{RISC-Linz,
  Johannes Kepler University, Linz, Austria}.
\newblock


\bibitem[\protect\citeauthoryear{Raab}{Raab}{2012b}]%
        {Raab2012}
\bibfield{author}{\bibinfo{person}{Clemens~G. Raab}.}
  \bibinfo{year}{2012}\natexlab{b}.
\newblock \showarticletitle{Using {G}r\"{o}bner bases for finding the
  logarithmic part of the integral of transcendental functions}.
\newblock \bibinfo{journal}{\emph{Journal of Symbolic Computation}}
  \bibinfo{volume}{47} (\bibinfo{year}{2012}), \bibinfo{pages}{1290--1296}.
\newblock


\bibitem[\protect\citeauthoryear{Raab and Singer}{Raab and Singer}{2022}]%
        {RS2022}
\bibfield{author}{\bibinfo{person}{Clemens~G. Raab} {and}
  \bibinfo{person}{Michael~F. Singer}.} \bibinfo{year}{2022}\natexlab{}.
\newblock \bibinfo{booktitle}{\emph{Integration in {F}inite {T}erms:
  {F}undamental {S}ources}}.
\newblock \bibinfo{publisher}{Springer Cham}, \bibinfo{address}{Switzerland}.
\newblock


\bibitem[\protect\citeauthoryear{Risch}{Risch}{1969}]%
        {Risch1969}
\bibfield{author}{\bibinfo{person}{Robert~H. Risch}.}
  \bibinfo{year}{1969}\natexlab{}.
\newblock \showarticletitle{The problem of integration in finite terms}.
\newblock \bibinfo{journal}{\emph{Trans. Amer. Math. Soc.}}
  \bibinfo{volume}{139} (\bibinfo{year}{1969}), \bibinfo{pages}{167--189}.
\newblock


\bibitem[\protect\citeauthoryear{Rosenlicht}{Rosenlicht}{1976}]%
        {Rosenlicht1976}
\bibfield{author}{\bibinfo{person}{Maxwell Rosenlicht}.}
  \bibinfo{year}{1976}\natexlab{}.
\newblock \showarticletitle{{On Liouville's theory of elementary functions}}.
\newblock \bibinfo{journal}{\emph{Pacific J. Math.}} \bibinfo{volume}{65},
  \bibinfo{number}{2} (\bibinfo{year}{1976}), \bibinfo{pages}{485 -- 492}.
\newblock


\bibitem[\protect\citeauthoryear{Salvy}{Salvy}{2019}]%
        {Salvy2019}
\bibfield{author}{\bibinfo{person}{Bruno Salvy}.}
  \bibinfo{year}{2019}\natexlab{}.
\newblock \showarticletitle{{Linear differential equations as a
  data-structure}}.
\newblock \bibinfo{journal}{\emph{Found.\ Comput.\ Math}} \bibinfo{volume}{19},
  \bibinfo{number}{5} (\bibinfo{year}{2019}), \bibinfo{pages}{1071--1112}.
\newblock


\bibitem[\protect\citeauthoryear{Schneider}{Schneider}{2021}]%
        {Schneider2021}
\bibfield{author}{\bibinfo{person}{Carsten Schneider}.}
  \bibinfo{year}{2021}\natexlab{}.
\newblock \showarticletitle{Term algebras, canonical representations and
  difference ring theory for symbolic summation}.
\newblock In \bibinfo{booktitle}{\emph{Anti-Differentiation and the Calculation
  of {F}eynman Amplitudes}}. \bibinfo{publisher}{Springer, Cham},
  \bibinfo{address}{Berlin}, \bibinfo{pages}{423--485}.
\newblock
\showISBNx{978-3-030-80218-9; 978-3-030-80219-6}


\bibitem[\protect\citeauthoryear{Schneider}{Schneider}{2023}]%
        {Schneider2023}
\bibfield{author}{\bibinfo{person}{Carsten Schneider}.}
  \bibinfo{year}{2023}\natexlab{}.
\newblock \showarticletitle{Refined telescoping algorithms in
  {$R\Pi\Sigma$}-extensions to reduce the degrees of the denominators}. In
  \bibinfo{booktitle}{\emph{Proceedings of the 48th {I}nternational {S}ymposium
  on {S}ymbolic and {A}lgebraic {C}omputation (ISSAC'23)}}.
  \bibinfo{publisher}{ACM}, \bibinfo{address}{New York, NY, USA},
  \bibinfo{pages}{498--507}.
\newblock
\showISBNx{9798400700392}


\bibitem[\protect\citeauthoryear{Singer, Saunders, and Caviness}{Singer
  et~al\mbox{.}}{1985}]%
        {SSC1985}
\bibfield{author}{\bibinfo{person}{Michael~F. Singer},
  \bibinfo{person}{David~B. Saunders}, {and} \bibinfo{person}{B.~F. Caviness}.}
  \bibinfo{year}{1985}\natexlab{}.
\newblock \showarticletitle{An Extension of Liouville's Theorem on Integration
  in Finite Terms}.
\newblock \bibinfo{journal}{\emph{SIAM J. Comput.}} \bibinfo{volume}{14},
  \bibinfo{number}{4} (\bibinfo{year}{1985}), \bibinfo{pages}{966--990}.
\newblock


\bibitem[\protect\citeauthoryear{{v}an {d}{e}{r}~Hoeven}{{v}an
  {d}{e}{r}~Hoeven}{2021}]%
        {vdHJ2021}
\bibfield{author}{\bibinfo{person}{Joris {v}an {d}{e}{r}~Hoeven}.}
  \bibinfo{year}{2021}\natexlab{}.
\newblock \showarticletitle{Constructing reductions for creative telescoping:
  the general differentially finite case}.
\newblock \bibinfo{journal}{\emph{Applicable Algebra in Engineering,
  Communication and Computing}} \bibinfo{volume}{32}, \bibinfo{number}{5}
  (\bibinfo{year}{2021}), \bibinfo{pages}{575--602}.
\newblock
\showISSN{0938-1279}


\bibitem[\protect\citeauthoryear{Zeilberger}{Zeilberger}{1990}]%
        {Zeil1990}
\bibfield{author}{\bibinfo{person}{Doron Zeilberger}.}
  \bibinfo{year}{1990}\natexlab{}.
\newblock \showarticletitle{A holonomic systems approach to special functions
  identities}.
\newblock \bibinfo{journal}{\emph{J. Comput. Appl. Math.}}
  \bibinfo{volume}{32}, \bibinfo{number}{3} (\bibinfo{year}{1990}),
  \bibinfo{pages}{321--368}.
\newblock
\showCODEN{JCAMDI}
\showISSN{0377-0427}


\end{thebibliography}

\appendix
\section{Empirical results}

We present some empirical results about in-field integration obtained by 
our complete reduction ({\tt CR}), Algorithm {\sc AddDecompInField} in~\cite[page 150]{DGLW2020} ({\tt AD}),
and the {\sc Maple} function  {\tt int}. 
Experiments were carried out with {\sc Maple 2021} on a computer with imac CPU 3.6GHZ, Intel Core i9, 16G memory.
{\sc Maple} scripts of {\tt CR} and {\tt AD} are available at 
\url{http://mmrc.iss.ac.cn/~zmli/ISSAC2025.html}. 

Every integrand in experimental data was a derivative in the primitive tower $\bQ(x)(t_1, t_2, t_3)$, 
where  $t_1=\log(x), t_2 = \log(x+1)$ and $t_3 = \log(t_1)$. So {\tt CR}, {\tt AD} and {\tt int} 
are all applicable and have the same output, which is an integral of the input in the same tower. 
Three integrands in the form $p_i^\prime$ were generated for each $i$,
where  $p_i$ was a dense polynomial in some selected generators. % with (total) degree~$i$. 
Below is a summary of the average timings (in seconds).

In the first suite of data, we set $p_i \in \bQ(x, t_1, t_2)[t_3]$ such that $\deg_{t_3}(p_i)=i$ and all coefficients of $p_i$ are rational functions
whose numerators and denominators are both sparse random polynomials in $\bQ[x, t_1, t_2]$ with total degree 5.

\noindent 
\begin{center}
\begin{tabular}{|c|c|c|c|c|c|c|c|} \hline
    $i$ &   $1$ & $2$ & $3$ &  $4$   & $5$ & $6$ \\ \hline
    {\tt  CR}    & 1.42 & 8.32& 37.01& 122.55&1085.04 & >3600 \\ \hline
    {\tt AD}      &  0.96 & 10.42 &47.36 &149.02 &>3600 & >3600 \\ \hline
    {\tt int}      & 1.15 &4.52&23.30&53.43&166.27&346.29\\ \hline
\end{tabular}
\end{center}

In the second suite,  $p_i \in \bQ(x, t_1, t_2)[t_3]$ with degree $i$ in $t_3$. 
Its coefficients are quotients of linear polynomials in $\bQ[x, t_1, t_2]$.

\noindent 
\begin{center}
\begin{tabular}{|c|c|c|c|c|c|c|c|} \hline
    $i$ &   $6$ & $8$ & $10$ &  $12$   & $14$ & $16$\\ \hline
    {\tt CR}    & 0.90& 2.09 & 7.05 & 12.56 & 30.35& 62.11\\ \hline
    {\tt AD}      &  1.23& 4.29& 12.31& 31.08& 57.67& 170.70\\ \hline
    {\tt int}     & 3.83& 17.46& 31.61& 66.22& 144.70& 322.19 \\ \hline
\end{tabular}
\end{center}

\noindent 
In the third suite,  $p_i \in \bQ(x)[t_1, t_2, t_3]$  whose total degree is equal to $i$ and whose 
coefficients are quotients of random polynomials in $\bQ[x]$ with degree 5.

\noindent 
\begin{center}
\begin{tabular}{|c|c|c|c|c|c|c|c|} \hline
    $i$ &   $1$ & $2$ & $3$ &  $4$   & $5$ & $6$\\ \hline
    {\tt CR}   & 0.35& 0.19& 0.59& 4.02& 21.32& 88.51 \\ \hline
    {\tt AD}      &  0.39& 0.51& 3.48 & 30.53& 614.90& 1453.61 \\ \hline
    {\tt int}      & 0.53& 0.63 & 4.68 & 51.82 & 154.31 & 1255.49\\ \hline
\end{tabular}
\end{center}

\noindent 
In the last suite, $p_i \in \bQ[x, t_1, t_2, t_3]$  with total degree $i$. 
The {\sc Maple} function {\tt int} returned expressions involving unevaluated integrals
for some inputs. Whenever this happened, 
the corresponding entry is marked by $\int$.
\begin{center}
\begin{tabular}{|c|c|c|c|c|c|c|c|} \hline
    $i$ &   $5$ & $10$ & $15$ &  $20$   & $25$ & $30$\\ \hline
    {\tt CR}   & 0.39& 0.25& 0.81& 1.98& 4.32& 8.71\\ \hline
    {\tt AD}      &  0.45& 1.06 & 6.69& 32.83 & 141.09& 280.47\\ \hline
    {\tt int}      & 0.49 & $\int$  & $\int$ &7.09&$\int$ & $\int $ \\ \hline
\end{tabular}
\end{center}

The timings reveal that {\tt CR} outperformed {\tt AD}, and was more efficient than {\tt int} except for the integrands in the first suite.
There are also examples for which {\tt int} took more than one hour without any output, but both {\tt CR} and {\tt AD} returned correct results. 

We also observe that {\sc HermiteReduce} and {\sc AuxiliaryReduction} were much more time-consuming than {\sc Projection}
in the complete reduction (see Algorithm \ref{ALG:cr}).

\end{document}